\def\year{2022}\relax
\documentclass[letterpaper]{article} 
\usepackage{aaai22}  
\usepackage{times}  
\usepackage{helvet}  
\usepackage{courier}  
\usepackage[hyphens]{url}  
\usepackage{graphicx} 
\usepackage{natbib}  
\usepackage{caption} 
\DeclareCaptionStyle{ruled}{labelfont=normalfont,labelsep=colon,strut=off} 
\usepackage{amsmath,amsthm,amssymb,algpseudocode,graphicx,tikz,caption,subcaption,xcolor,xspace,cleveref,nicefrac}
\definecolor{light-gray}{gray}{0.7}

\usepackage[shortlabels]{enumitem}

\usepackage[ruled,vlined,linesnumbered]{algorithm2e}

\algrenewcommand\algorithmicindent{0.6em}
\newtheorem{theorem}{Theorem}[section]
\newtheorem{proposition}[theorem]{Proposition}
\newtheorem{lemma}[theorem]{Lemma}

\newtheorem*{stmnt*}{Statement}
\algnotext{EndFor}
\algnotext{EndIf}
\algnotext{EndWhile}
\algnotext{EndProcedure}

\newcommand{\p}{{{\mathrm{P}}}}
\newcommand{\np}{{{\mathrm{NP}}}}

\theoremstyle{definition}
\newtheorem{definition}[theorem]{Definition}
\newtheorem{example}[theorem]{Example}
\newtheorem{corollary}[theorem]{Corollary}

\usepackage{newfloat}
\usepackage{listings}
\newcommand{\calA}{\mathcal{A}}

\newcommand{\calS}{\mathcal{S}}
\newcommand{\calP}{\mathcal{P}}
\newcommand{\sw}{\mathsf{sw}}

\newcommand{\cov}{\mathsf{cvr}}

\newcommand{\JR}{\mathrm{JR}}
\newcommand{\EJR}{\mathrm{EJR}}

\newcommand{\uncov}{\mathsf{uncov}}
\newcommand{\setcov}{\mathsf{setcov}}
\def\R{\mathbb{R}}

\DeclareMathOperator{\opt}{OPT}
\DeclareMathOperator{\eds}{EDS}

\nocopyright
\allowdisplaybreaks

\title{The Price of Justified Representation}
\author {
    Edith Elkind,\textsuperscript{\rm 1}
    Piotr Faliszewski,\textsuperscript{\rm 2}
    Ayumi Igarashi,\textsuperscript{\rm 3} 
    Pasin Manurangsi,\textsuperscript{\rm 4} \\
    Ulrike Schmidt-Kraepelin,\textsuperscript{\rm 5}
    Warut Suksompong\textsuperscript{\rm 6}
}

\affiliations {
    \textsuperscript{\rm 1} University of Oxford,
    \textsuperscript{\rm 2} AGH University of Science and Technology,
    \textsuperscript{\rm 3} National Institute of Informatics \\
    \textsuperscript{\rm 4} Google Research 
    \textsuperscript{\rm 5} TU Berlin 
    \textsuperscript{\rm 6} National University of Singapore \\
}

\begin{document}

\maketitle

\begin{abstract}
In multiwinner approval voting, the goal is to select $k$-member committees based on voters' approval ballots. A well-studied concept of proportionality in this context is the {\em justified representation (JR)} axiom, which demands that no large cohesive group of voters remains unrepresented. However, the JR axiom may conflict with other desiderata, such as {\em coverage} (maximizing the number of voters who approve at least one committee member) or {\em social welfare} (maximizing the number of approvals obtained by committee members). In this work, we investigate the impact of imposing the JR axiom (as well as the more demanding EJR axiom) on social welfare and coverage. Our approach is threefold: we derive worst-case
bounds on the loss of welfare/coverage that is caused by imposing JR, study the computational complexity of finding `good' committees that provide JR (obtaining a hardness result, an approximation algorithm, and an exact algorithm for one-dimensional preferences), and examine this setting empirically on several synthetic datasets.
\end{abstract}

\section{Introduction}
\label{sec:intro}
What do the tasks of electing members of a parliament, making a shortlist of job candidates to invite for an interview, and selecting dishes for a banquet have in common? They can all be seen
as instances of multiwinner voting: there is a set of voters with preferences over the candidates, and the goal is to select a fixed-size subset of the candidates based on these 
preferences. 
Recently, axiomatic and algorithmic questions related to multiwinner voting have attracted significant attention from the AI community \citep{FSST17}.

A common form of multiwinner voting involves \emph{approval ballots}, wherein each voter specifies the subset of candidates that she finds acceptable.
Depending on the application, one may want the winning set (also called a {\em committee}) to maximize the {\em utilitarian social welfare} (i.e., to select the candidates with the largest number of approvals)
or {\em coverage} (i.e., to maximize the number of voters
who approve at least one of the selected candidates).
Another objective, which is somewhat more difficult to capture, is to select
a committee that represents the set of voters in a proportional manner.
A fairly basic proportionality requirement is the \emph{justified
representation (JR)} axiom, formulated by \citet{ABC17}: it says that in an $n$-voter election where the 
goal is to select a committee of size $k$, no group of $\nicefrac{n}{k}$ voters who jointly
approve some candidate should remain unrepresented in the elected committee. 

\paragraph{Our Contribution}
In this paper, we seek to understand the effect of imposing the JR axiom on the
objectives of social welfare and coverage. 
We observe that requiring a committee to provide JR has no impact
on coverage: there is a committee that offers optimal coverage as well as JR. In contrast, the impact on welfare is substantial: an example of \citet{LS20} shows that
imposing the JR axiom may result in a committee whose social
welfare differs from optimal by a factor of $\nicefrac{\sqrt{k}}{2}$, and we prove that this bound is tight; in fact, we can obtain a nearly $\nicefrac{\sqrt{k}}{2}$-approximation to optimal social welfare, while providing both JR and a $\nicefrac43$-approximation to optimal coverage. 
For coverage, we also show that the loss caused by imposing a more demanding axiom of {\em extended justified representation} is bounded by $\nicefrac43$. 

We then investigate the complexity of computing a wel\-fare-maximizing JR committee. While both the problem of computing a JR committee and the problem of maximizing the social welfare are in $\p$, combining these two objectives results in a problem that is $\np$-hard, even to approximate up to a factor of $k^{\nicefrac12-\varepsilon}$ for $\varepsilon>0$. 
On the positive side, we show that  it can be efficiently approximated up to a factor of nearly $\nicefrac{\sqrt{k}}{2}$, which means that the inapproximability bound is asymptotically tight.
Moreover, we present a polynomial-time exact algorithm for one-dimensional preferences (namely, for the fairly large
1D-VCR domain, recently proposed by \citet{GBSF21}).

We complement our theoretical findings by an empirical analysis. We consider three standard models of preference distributions, and estimate the impact of imposing the JR axiom on welfare and coverage. In contrast to our theoretical results, this analysis paints
a much more optimistic picture: one can achieve a good approximation to the Pareto frontier of social welfare and coverage even under the JR constraint.

\paragraph{Related Work}
There is a quickly growing body of work on approval-based multiwinner voting; we point the reader to the survey of \citet{LS-survey}. Our work builds on several predecessor papers, 
which study trade-offs among different desiderata in the context of committee selection, but take a somewhat different perspective. 

\citet{LS20} consider
a number of approval-based multiwinner voting rules and, for each
of them, establish theoretical worst-case guarantees on its performance with respect to social welfare and coverage. They also provide some experimental results, showing that in practice these rules often perform much better than their guarantees suggest. The main
difference between their work and ours is that they consider 
specific voting rules (including some that provide JR), whereas we analyze the impact of the JR axiom itself. Furthermore, \citet{LS20} take the worst-case approach, i.e., they measure the performance of a rule by the {\em minimum} welfare/coverage provided by committees that it outputs, whereas we are interested in the {\em maximum} welfare/coverage that can be obtained by committees that provide JR.
Nevertheless, some of their results and examples turn out to be relevant for our analysis.

\citet{KKE+19a} consider
ordinal elections rather than the approval ones and analyze the complexity of finding a specific committee that achieves good scores according to several different rules. In particular, they consider the problem of maximizing  the $k$-Borda score subject to achieving a given Chamberlin--Courant score~\citep{CC83}. This is quite similar to our problem of finding committees that maximize social welfare subject to JR. Our plots illustrating the trade-off between the social welfare and the coverage that JR committees may provide are inspired by analogous plots of \citet{KKE+19a}.

\citet{BFKN19} initiate the study of committees that provide JR.
In particular, they show that maximizing welfare or coverage subject to JR is computationally hard; we strengthen their $\np$-hardness result for welfare to a $k^{\nicefrac12-\varepsilon}$-inapproximability result. They also conduct experiments in which they study trade-offs between 
coverage and social welfare, with and without imposing the JR axiom. However, for the purposes of their analysis, they only consider three specific elections, whereas we sample many elections from three different distributions; thus, our empirical results are much more robust.

\section{Preliminaries}
\label{sec:prelim}

We consider elections with a finite set of candidates $C$ of size $m$ and 
a finite set of voters $N=[n]$, where we write $[t] := \{1,\dots,t\}$ for any positive integer $t$.
Each voter $i\in N$ submits a non-empty ballot
$A_i\subseteq C$, and the goal is to select a subset of $C$ of size $k$, 
which we will refer to as the {\em winning committee}. Thus, an instance $I$ of our problem 
can be described by a set of candidates $C$, 
a list of ballots $\calA=(A_1, \dots, A_n)$, and a positive integer 
$k$; we write $I=(C, \calA, k)$. 
If $c\in A_i$ for some $c\in C$ and $i\in N$, we say that $c$ {\em covers}
$i$, or, equivalently, that $i$ {\em approves} $c$.

We consider the following two measures of committee quality:
\begin{enumerate}
\item (Utilitarian) social welfare: For each committee $W\subseteq C$,
  we define the {\em (utilitarian) social welfare} of $W$ as
$$
\sw(W)=\sum_{i\in N}|A_i\cap W|.
$$
\item
Coverage: For each committee $W\subseteq C$, 
we define the {\em coverage} of $W$ as
$$
\cov(W)=|\{i\in N: A_i\cap W\neq\varnothing\}|.
$$
\end{enumerate}

Given an instance $I = (C, \calA, k)$ with $\calA=(A_1, \dots, A_n)$,  
we say that a group of voters $N'\subseteq N$ is {\em cohesive} if
$\cap_{i\in N'} A_i\neq\varnothing$. Further, we say that a committee $W$
{\em represents} a group of voters $N'\subseteq N$
if $W\cap A_i\neq\varnothing$ for some $i\in N'$. We are now ready to state
the justified representation axiom of \citet{ABC17}.

\begin{definition}[JR]
\label{def:jr}
Given an instance $I = (C, \calA, k)$ with $\calA=(A_1, \dots, A_n)$,  
we say that a committee $W\subseteq C$, $|W|=k$, provides 
{\em justified representation (JR)}
for $I$ if it represents 
every cohesive group of voters $N'\subseteq N$ 
such that $|N'|\ge \nicefrac{n}{k}$. 
Let $\JR(I)$ be the set of all committees 
that provide justified representation for $I$.
\end{definition}

Fix an instance $I=(C, \calA, k)$. 
We write 
\begin{align*}
&\sw(I) =\max_{\substack{W\subseteq C\\ |W|=k}}\sw(W), \quad
\sw_\JR(I)=\max_{W\in\JR(I)}\sw(W),\\
&\cov(I)=\max_{\substack{W\subseteq C\\ |W|=k}}\cov(W),\quad
\cov_\JR(I)=\max_{W\in\JR(I)}\cov(W),\\
&P_\sw(I)=\frac{\sw(I)}{\sw_\JR(I)}, \quad
P_\cov(I)=\frac{\cov(I)}{\cov_\JR(I)}.
\end{align*}

We refer to $P_\sw(I)$ and $P_\cov(I)$ as the 
{\em social welfare price of JR on $I$} and
the {\em coverage price of JR on $I$}, respectively.
Given a committee size $k$, 
the \emph{social welfare price of JR} 
and the \emph{coverage price of JR} are defined as
the largest values $P_\sw$ and $P_\cov$ 
can take on instances with committee size $k$:
$$
\calP_\sw(k)=\sup_{I=(C, \calA, k)}P_\sw(I), \quad
\calP_\cov(k)=\sup_{I=(C, \calA, k)}P_\cov(I).
$$

We note that the sheer fact that a committee provides JR does not
imply anything about its utilitarian social welfare or
coverage. Indeed, there exist instances where the JR axiom is
non-binding, in the sense that there are no cohesive groups of size $\nicefrac{n}{k}$.  For such instances even a committee
that consists of candidates not approved by {\em any} voter provides
JR. This is why we defined the social welfare price (the coverage
price) of JR in terms of the maximum welfare (coverage) achievable by
a JR committee and not the minimum one.

In our proofs, we will often refer to a family of algorithms known as
GreedyCC. Each algorithm in this family takes an instance $I=(C, \calA, k)$ as input, outputs
a committee $W\subseteq C$ of size $k$, and proceeds in two stages. In the first stage, 
it adds candidates to $W$ one by one, starting with $W=\varnothing$.
At each iteration, for each candidate $c\not\in W$ it computes 
$\cov(W\cup\{c\})-\cov(W)$; let $c$ be a candidate that maximizes this quantity.
If $\cov(W\cup\{c\})-\cov(W)\ge \nicefrac{n}{k}$ then the algorithm adds
$c$ to the committee and proceeds to the next iteration. The first stage
ends when $\cov(W\cup\{c\})-\cov(W)<\nicefrac{n}{k}$ for all $c\not\in W$;
let $W'$ be the committee obtained at this point, and let $k'=|W'|$. 
As each candidate added in the first stage covers 
$\nicefrac{n}{k}$ `fresh' voters, we have $k'\le k$, 
and we have $W\in\JR(I)$ for any size-$k$ committee $W$ with $W'\subseteq W$.
In the second stage, $k-k'$ further
candidates are added to the committee. Specifically, we denote by GreedyCC$^\sw$
the variant of GreedyCC that adds $k-k'$ candidates with the highest
number of approvals among candidates in $C\setminus W'$ (breaking ties in favor of lower-indexed candidates).

\section{Price of Justified Representation}\label{sec:price}

In this section we provide nearly tight bounds
on $\calP_\sw(k)$ and $\calP_\cov(k)$.
\citet{LS20} give an example showing 
that $\calP_\sw(k)\ge \nicefrac{\sqrt{k}}{2}$
(see Example~\ref{ex:av} in Appendix~\ref{app:price});
we provide an upper bound on $\calP_\sw(k)$ 
that (nearly) matches this lower bound. 
On the other hand, we observe that $\calP_\cov(k)=1$.
Then we show that social welfare and coverage are
surprisingly compatible when considering JR committees. 
We also discuss 
how to extend our analysis to extended justified representation (EJR).

\subsection{Social Welfare}
\label{sec:welfare}

We first consider $\calP_\sw(k)$. In addition to the lower 
bound of $\nicefrac{\sqrt{k}}{2}$ (Example~\ref{ex:av}),
\citet{LS20} also show that the social welfare 
of the committees output by the
\emph{Proportional Approval Voting} rule (PAV), 
which is well-known to provide
JR~\cite{ABC17}, is within a factor of
$\sqrt{k}+2$ from the optimal social welfare;
this implies that $\calP_\sw(k)\le \sqrt{k}+2$.
We improve this upper bound to one that, 
in essence, matches the bound from Example~\ref{ex:av}.

\begin{theorem}\label{thm:overall-price-sw}
  For each positive number $\beta < 2$, there is a committee size $k'$
  such that for all $k\ge k'$, we have 
  $\calP_\sw(k)\le \frac{1}{\beta}(1+\sqrt{k})$.
\end{theorem}
\begin{proof}
  Consider an $n$-voter instance $I = (C, \calA, k)$. 
  Let $W$ be the committee computed by
  GreedyCC$^\sw$ (so that $W\in\JR(I)$), and let $S = \{s_1, \ldots, s_k\}$ be
  a committee such that $\sw(S)=\sw(I)$. For each $i\in [k]$, let 
  $a_i$ be the number of approvals received by candidate $s_i$.
  We can assume without loss of generality that
  $a_1 \geq a_2 \geq \cdots \geq a_k$. For each $i\in [k]$, 
  let $S^i=\{s_1, \dots, s_i\}$ be the set of $i$ candidates 
  with the highest approval scores; note that $\sw(S^i)\ge \frac{i}{k}\cdot\sw(S)$.
  By definition, some candidate with $a_1$ approvals is selected in the first
  iteration of GreedyCC$^\sw$;
  we can assume without loss of generality that this candidate is $s_1$.

  Suppose first that GreedyCC$^\sw$ selects at most $k - \lfloor 2\sqrt{k} \rfloor$ candidates during the first stage.
  Consider a candidate $s_i$, $i\le \lceil 2\sqrt{k} \rceil$.
  If $s_i$ is not selected during the first stage,
  then $s_i$ 
  will  necessarily be selected during the second stage.\footnote{We bound $i$ by $\lceil 2\sqrt{k}\rceil$ and not $\lfloor 2\sqrt{k}\rfloor$ because we know that $s_1$ was selected in the first
  stage.}
  Hence, we have: 
  \[
     \sw(W) \geq (\nicefrac{\lceil 2\sqrt{k}\rceil}{k})\cdot \sw(S) \geq (\nicefrac{2}{\sqrt{k}})\cdot\sw(S).
  \]
  So, in this case
  $P_\sw(I)\le \nicefrac{\sqrt{k}}{2}\le \frac{1+\sqrt{k}}{\beta}$ for each $\beta<2$.
  Thus let us assume that GreedyCC$^\sw$ selects $k - x\sqrt{k}$
  first-stage candidates for some $0 \leq x < 2$. Let $r=x\sqrt{k}$;
  note that $r$ is an integer. We will assume that $k\ge 2r+1$; this is true for large enough values of $k$. 
  
  Recall that in the first iteration, GreedyCC$^\sw$ chooses candidate $s_1$, who is approved by $a_1$ voters. Afterwards, it selects
  $k - r -1$ further first-stage candidates, with each of these candidates covering at least $\nicefrac{n}{k}$ additional voters. Hence, it must be the case that:
  \[ \textstyle
     n - a_1 \ge \frac{n}{k}\left(k - r -1\right).
  \]
  Recalling that $r=x\sqrt{k}$, we obtain $a_1 \leq \frac{n}{k} + \frac{xn}{\sqrt{k}}$
  and hence $a_i\leq \frac{n}{k} + \frac{xn}{\sqrt{k}}$ for $i\in [r]$.
  On the other hand, since each of the $k-r-1\ge r$ candidates selected in the first stage besides $s_1$ covers at least $\frac{n}{k}$ additional voters, for each $i \in [r]$ we have $a_i \geq \frac{n}{k}$. Hence, 
  there exists a real value $\alpha \in [0,1]$
  such that the average number of approvals obtained by candidates in $S^r$ is:
  \[ \textstyle
     \frac{1}{r}\sum_{i=1}^{r} a_i =  \frac{n}{k} +\frac{\alpha xn}{\sqrt{k}}.
  \] 
  Thus the social welfare of $S$ can be upper-bounded as:
  \[ \textstyle
    \sw(S) \leq k \cdot \left(\frac{n}{k} +\frac{\alpha xn}{\sqrt{k}}\right) = n(1+ \alpha x \sqrt{k}). 
  \]
  We will now establish a lower bound on the social welfare of $W$. 
  Since we add $r$ candidates in the second stage, 
  $W$ necessarily contains all candidates in $S^r$, who contribute a total of at least $r(\frac{n}{k}+\frac{\alpha xn}{\sqrt{k}})$ to the social welfare. Moreover, in the first stage, we add at least $k-r-|S^r|=k-2r$ candidates who are not contained in $S^r$.
  Each of the first-stage candidates 
  covers at least $\nicefrac{n}{k}$ additional voters and therefore contributes at least $\nicefrac{n}{k}$ to the social welfare. Hence, we can lower-bound $\sw(W)$ as:
  \begin{align*}
      \textstyle
      \sw(W) \geq \left(\frac{n}{k}+\frac{\alpha x n}{\sqrt k}\right)\cdot r + (k - 2r)\cdot\frac{n}{k}, 
  \end{align*}
  or, rewriting, 
  \[ \textstyle
    \sw(W) \geq n - \frac{xn}{\sqrt{k}} + \alpha x^2 n = n\left(1 - \frac{x}{\sqrt{k}} + \alpha x^2\right).
  \]
  Now, $P_\sw(I)$ can be upper-bounded as:
  \begin{align*}\textstyle
      \frac{\sw(S)}{\sw(W)} \leq \frac{1+ \alpha x\sqrt{k}}{1 - \frac{x}{\sqrt{k}} + \alpha x^2} =
      \frac{\sqrt{k}+ \alpha x k}{\sqrt{k} - x + \alpha x^2 \sqrt{k}}.
  \end{align*}
  We claim that for each $\beta \in (0,2)$ and sufficiently large $k$
  it holds that:
  \[  \textstyle
      \frac{\sqrt{k}+ \alpha x k}{\sqrt{k} - x + \alpha x^2 \sqrt{k}} \leq \frac{1+\sqrt{k}}{\beta}.
  \]
  Note first that this inequality can be rewritten as:
  \[\textstyle
     \sqrt{k} - x + \alpha x^2 \sqrt{k} + k -x\sqrt{k} + \alpha x^2 k \geq \beta \sqrt{k} + \alpha \beta x k,
  \]
  or, regrouping the terms and dividing by $\sqrt{k}$,  
  \begin{equation}
     \label{eq:sw:jr}
     \textstyle
     -1 - \alpha x^2 + x +\beta + \frac{x}{\sqrt{k}}   \leq \sqrt{k}\big(\alpha x^2 - \alpha \beta x  + 1 \big).
  \end{equation}
  Now, observe that
  \begin{align*}
  \alpha x^2 - \alpha \beta x + 1
  &= \left(\alpha x^2 - \alpha \beta x + \tfrac{\alpha\beta^2}{4}\right) + \left(1 - \tfrac{\alpha\beta^2}{4}\right) \\
  &= \alpha\left(x-\tfrac{\beta}{2}\right)^2 + \left(1 - \tfrac{\alpha\beta^2}{4}\right) \\
  &\ge 1 - \tfrac{\alpha\beta^2}{4} \ge 1 - \tfrac{\beta^2}{4} > 0.
  \end{align*}
  As a consequence, inequality~\eqref{eq:sw:jr} holds for sufficiently
  large~$k$ since the left-hand side does not increase as $k$ grows, whereas the right-hand side grows proportionally to $\sqrt{k}$.
\end{proof}

For an explicit bound on $k$ in the statement of Theorem~\ref{thm:overall-price-sw}, we refer to Appendix~\ref{app:explicit-bound}.
One may wonder if it is possible to strengthen
the theorem so that it matches the
$\nicefrac{\sqrt{k}}{2}$ lower bound exactly, for all values of
$k$. While it might be possible, it would likely require an approach that is different
from ours: the approximation of the social welfare
provided by GreedyCC$^\sw$ is quite good, but,
nonetheless, insufficient for small committees.

\subsection{Coverage}
\label{sec:coverage}

Next we consider the coverage price of justified
representation. Since the committees selected by the
Chamberlin--Courant rule are known to both provide JR
and maximize the coverage \citep{ABC17}, we have the following
corollary.

\begin{corollary}
  For every instance $I=(C, \calA, k)$, there exists a committee
  $W\in \JR(I)$ such that $\cov(W)\ge \cov(S)$ for all $S\subseteq C$
  with $|S|=k$. Thus, $\calP_\cov(k)=1$ for all $k\in\mathbb N$.
\end{corollary}

\subsection{Combining Social Welfare and Coverage}

Social welfare and coverage are often viewed as two desiderata that are at the opposite ends of the spectrum: maximizing the social
welfare may lead to low coverage and vice versa. We show that
for JR committees these two properties are surprisingly compatible: 
if the committee size is sufficiently large, 
there always exists a JR committee that provides nearly optimal social welfare
as well as coverage that is nearly $\nicefrac{3}{4}$ of the maximum possible coverage.

\begin{theorem}\label{thm:jr-swcov}
  For every $\beta \in (0, 2)$ and $\varepsilon>0$,
  there exists a $k_0\in\mathbb N$ such that for every instance
  $I = (C, \calA, k)$ with $k\ge k_0$, 
  there exists a committee $W \in \JR(I)$ with
  $\sw(I)/\sw(W)\le \frac{1}{\beta}(1+\sqrt{k})$ and
  $\cov(I)/\cov(W)\le \nicefrac43+\varepsilon$.
\end{theorem}

\begin{proof}
  Consider an instance $I = (C,\calA, k)$ with
  $n$ voters, and let $W$ be a committee computed
  using a variant of GreedyCC with the following second-stage policy:
  Let $W'$ be the committee selected in the first stage, 
  let $N'$ be the set of voters covered by $W'$, and let $k'=|W'|$. 
  In the second stage, we first choose a set of candidates $W''$ of size 
  $\max(k-k'-\lfloor2\sqrt{k}\rfloor, 0)$ so as to maximize the coverage of 
  $W'\cup W''$ (this step may require exponential
  time). Then we choose
  the remaining $k-|W'\cup W''|$ candidates so as 
  to maximize the social welfare.
  Let $W$ be the resulting committee.
  We consider two cases.
  
  \smallskip

  \noindent \textbf{Case 1:} $\boldsymbol{k-k' \leq \lfloor 2\sqrt{k} \rfloor}$. In this case, $W''=\varnothing$,
  so the committee $W$ is identical to the committee computed by GreedyCC$^\sw$.
  Thus, by Theorem~\ref{thm:overall-price-sw}, for sufficiently large $k$ the committee 
  $W$ meets the requirements regarding social
  welfare. Further, since $k'\ge k-\lfloor 2\sqrt{k} \rfloor$, i.e., 
  at least $k - \lfloor 2\sqrt{k} \rfloor$ candidates were selected during the first stage, 
  $W$ covers at least
  $(k-\lfloor 2\sqrt{k}\rfloor)\frac{n}{k} \geq n(1 - \frac{2}{\sqrt{k}})$ voters.  For
  $k \geq 64$ this value is at least $\nicefrac{3n}{4}$.
  This completes the analysis for the first case.

  \smallskip

  \noindent \textbf{Case 2:} $\boldsymbol{k-k' > \lfloor 2\sqrt{k}\rfloor}$.  Then, $|W\setminus(W'\cup W'')|=\lfloor 2\sqrt{k}\rfloor$. 
  Since the first candidate in $W'$ as well as the candidates in $W\setminus (W'\cup W'')$ 
  are selected so as to maximize the social welfare,
  as argued in the proof of Theorem~\ref{thm:overall-price-sw},
  the set $W$ contains the top $1+\lfloor 2\sqrt{k} \rfloor\ge 2\sqrt{k}$ candidates in $C$ with respect to
  the number of approvals. 
  Hence, $\sw(W)\ge \frac{2\sqrt{k}}{k}\cdot \sw(I) \ge \frac{\beta}{1+\sqrt{k}}\cdot \sw(I)$. 

  In what follows, we will focus on the coverage provided by $W$
  (and specifically by $W'\cup W''$).
  Pick a committee $O$ with $\cov(O)=\cov(I)$, and let $N_O$
  denote the set of voters covered by $O$. 
  Let $\alpha =\nicefrac{k'}{k}$ and  
      $\alpha'=\alpha+ \nicefrac{\lfloor 2\sqrt{k} \rfloor}{k}$. Then $|W''| = (1-\alpha')k$;
      recall that the candidates in $W''$ were selected to maximize coverage.
  
  Consider the set of voters $N'$ covered by $W'$. Since $W'$ 
  is chosen in the first stage of GreedyCC and $|W'|=k'$, 
  we have $|N'|\ge k'\cdot \nicefrac{n}{k}\ge \alpha\cdot \cov(I)$.
  Next, let $N''= N_O\setminus N'$ and consider the voters in $N''$. We claim that  
  there exists a subset $O'\subseteq O$ of size $(1-\alpha')k$ that covers at least 
  $(1-\alpha')|N''|$ voters in $N''$. Indeed, let $o_1, \dots, o_k$
  be an arbitrary ordering of candidates in $O$, and for each $i\in [k]$ 
  let $n_i$ be the marginal contribution of $o_i$ to the coverage of $N''$, i.e., 
  the number of voters in $N''$ that are covered by $o_i$, 
  but not by $o_1, \dots, o_{i-1}$. We have $n_1+\dots+n_k=|N''|$,
  so the set $O'$ consisting of $(1-\alpha')k$ candidates in $O$ with the largest values of $n_i$ 
  covers at least $(1-\alpha')|N''|$ voters
  in $N''$. 

  Let $\gamma= |N_O\cap N'|/|N_O|$. 
  Then $|N''|=(1-\gamma)|N_O|$, 
  i.e., it is possible to pick $(1-\alpha')k$ candidates in $O$
  to cover at least $(1-\alpha')(1-\gamma)|N_O|$ voters in $N''$.
  Hence, as $W''$ was chosen to maximize the coverage of voters in 
  $N\setminus N'$ (which contains $N''$),
  it covers at least $(1-\alpha')(1-\gamma)|N_O|$ voters in $N\setminus N'$.

  We will now argue that the candidates in $W'\cup W''$ cover at least
  $(\alpha + (1-\alpha')^2)\cov(I)$ voters.
  We consider two cases:
    \begin{enumerate}
  \item Suppose that $\gamma \leq \alpha'$. Then 
    candidates in $W''$ cover at least $(1-\alpha')(1-\gamma)|N_O| \geq (1-\alpha')^2|N_O|$
    voters in $N\setminus N'$. As $|N_O|=\cov(I)$ and
    $|N'|\ge \alpha\cdot\cov(I)$, our claim follows.
  \item Now suppose that $\gamma > \alpha'$. Since $\alpha'>\alpha$, we have:
    \[ 
       |N'|\ge |N_O\cap N'| = \gamma |N_O| \geq \alpha |N_O| + (\gamma-\alpha')|N_O|, 
    \] 
    whereas the number of voters in $N\setminus N'$ that are covered by $W''$ is at least:
    \begin{align*} 
    &(1-\alpha')(1-\gamma)|N_O| \\
    &= (1-\alpha')^2|N_O| - (\gamma - \alpha')(1-\alpha')|N_O|. 
    \end{align*} 
      As $N_O=\cov(I)$, the number of voters covered by $W'\cup W''$ is at least:
    \begin{align*} 
      &\alpha |N_O| + (1-\alpha')^2|N_O| \\
      &+ \underbrace{(\gamma-\alpha')|N_O| - (\gamma -
      \alpha')(1-\alpha')|N_O|}_{\geq 0}\\ &\geq (\alpha + (1-\alpha')^2)\cov(I).
    \end{align*}
  \end{enumerate}
  As $\alpha'-\alpha=\nicefrac{\lfloor 2\sqrt{k}\rfloor}{k}$ and $\alpha,\alpha'\in[0,1]$, 
  for every $\varepsilon>0$ there exists a $k_0\in\mathbb N$
  such that $(\alpha'-\alpha)(2-\alpha-\alpha') \le \nicefrac{\varepsilon}{4}$ for all $k\ge k_0$.
  When this happens, we have
  $\alpha+(1-\alpha')^2\ge \alpha+(1-\alpha)^2-\nicefrac{\varepsilon}{4}$ for all $k\ge k_0$.
  Further, for $\alpha \in [0,1]$, the expression $\alpha + (1-\alpha)^2$ 
  is minimized at $\alpha = \nicefrac{1}{2}$ and evaluates to $\nicefrac34$
  at that point. Therefore, for $k\ge k_0$ the number of voters covered by $W$
  is at least $\frac{3-\varepsilon}{4}\cdot \cov(I)$ and hence 
  $\cov(I)/\cov(W)\le \nicefrac{4}{(3-\varepsilon)} \le \nicefrac43+\varepsilon$, where the last inequality holds for any $\varepsilon\in(0,1]$. 
  The statement for $\varepsilon > 1$ follows from that for $\varepsilon = 1$.
\end{proof}

The proof of Theorem~\ref{thm:jr-swcov} relies on a variant
of GreedyCC that runs in exponential
time. This is inevitable since, unless $\p = \np$,
the best approximation ratio for 
coverage that is guaranteed in polynomial time is
$1 - \nicefrac{1}{e}$ (see, e.g., \citealt{SF17}), and
$1/(1-\nicefrac{1}{e}) > 4/3$. It is thus desirable
to have a variant of Theorem~\ref{thm:jr-swcov} for committees computable
in polynomial time. It turns out that in this case 
we can achieve a near-perfect result.

\begin{theorem}\label{thm:jr-swcov-poly}
  Let $\beta < 2$ be a positive number.  For every instance
  $I = (C, \calA, k)$ with sufficiently large $k$, there is a
  committee $W \in \JR(I)$ with
  $\sw(I)/\sw(W)\le \frac{1}{\beta}(1+\sqrt{k})$ and
  $\cov(I)/\cov(W)\le 1/(1-(\nicefrac{1}e)^{(1 - \nicefrac{2}{\sqrt{k}}-\nicefrac{1}{k})})$.
\end{theorem}

\begin{proof}
  Given an instance $I=(C, \calA, k)$, we modify the variant of GreedyCC described in the proof of Theorem~\ref{thm:jr-swcov} as follows:
  instead of choosing $W''$ optimally, we choose it greedily, i.e., we add candidates
  in $W''$ one by one so as to maximize coverage at each iteration. Clearly, 
  this modification of GreedyCC runs in polynomial time.

  The bound 
  on the social welfare of the resulting committee follows by the arguments in the
  proof of Theorem~\ref{thm:jr-swcov}. For coverage, observe that,
  in effect, we run the greedy covering algorithm for at least  
  $k-\lfloor 2\sqrt{k}\rfloor\ge k-2\sqrt{k}-1$ iterations. 
  It is well-known that running $\ell$ iterations of this algorithm provides 
  a $(1-(\nicefrac{1}{e})^{\nicefrac{\ell}{k}})$-approximation to optimal coverage; 
  see the original work of~\citet{nem-wol-fis:j:submodular} or the overview of
  \citet{kra-gol:b:submodular}. Substituting $\ell \ge k-2\sqrt{k}-1$
  implies that our committee covers at
  least $(1-(\nicefrac{1}e)^{(1 - \nicefrac{2}{\sqrt{k}}-\nicefrac{1}{k})})\cdot \cov(I)$
  voters. 
\end{proof}

\subsection{Extended Justified Representation}
\label{sec:ejr}

We also study a strengthening of the JR axiom, known as 
{\em extended justified representation (EJR)} \citep{ABC17}, which is
typically interpreted as a proportionality axiom. 
Let us consider an election with $n$
voters, where we seek a committee of size $k$. We say that a group of
voters is {\em $\ell$-cohesive} if there are at least $\ell$ candidates that
are approved by all the voters in this group, and we say that this group is
{\em $\ell$-large} if it contains at least $\nicefrac{\ell n}{k}$ members.

\begin{definition}[EJR]
\label{def:ejr}
Given an instance $I = (C, \calA, k)$,
we say that a committee $W\subseteq C$, $|W|=k$, provides {\em
  extended justified representation (EJR)} for $I$ if for each
$\ell \in [k]$ and each $\ell$-cohesive, $\ell$-large group of voters,
there is at least one voter in this group who approves at least $\ell$ members
of $W$.  Let $\EJR(I)$ denote the set of all committees that provide
extended justified representation for $I$.
\end{definition}
We define the social welfare price of EJR and the coverage
price of EJR analogously to these notions for JR. The results of \citet{LS20}
imply that the social welfare price of EJR is between
$\nicefrac{\sqrt{k}}{2}$ and $2 + \sqrt{k}$, as well as that the coverage price of EJR
is at least $\nicefrac{4}{3}$. 
We show that, in fact, this latter bound is tight.

\begin{theorem}\label{thm:ejr}
  The coverage price of EJR is $\nicefrac{4}{3}$.
\end{theorem}

To prove Theorem~\ref{thm:ejr}, we start by describing a greedy algorithm for computing a committee that provides EJR.
This algorithm, which we will refer to as \emph{GreedyEJR}, is presented
as Algorithm~\ref{alg:greedy-ejr}. It is similar in spirit to GreedyCC, 
but its running time may be exponential in the instance size.

Briefly put, the algorithm considers the
numbers $\ell$ from $k$ down to $1$ and for each $\ell$ it greedily
considers $\ell$-cohesive, $\ell$-large groups (starting with the
largest ones); for each such group, it extends the committee with
$\ell$ candidates that are approved by each member of the group; we refer
to these candidates as \emph{first-stage} committee members. If 
at the end of this process fewer than $k$ candidates are selected, 
it fills in the committee arbitrarily (see the final \emph{if} statement in
Algorithm~\ref{alg:greedy-ejr}); we refer to these new committee
members as \emph{second-stage} candidates.

GreedyEJR was first used within the proof of Theorem~3 of \citet{BFKN19}\footnote{Unfortunately, the proof of this theorem is
  not provided in the conference version of their paper and the full
  version is not publicly available.} and recently its extension was
considered by \citet{PPS20}.
In their unpublished proof, \citet{BFKN19} argued that the committees
returned by GreedyEJR provide EJR, and \citet{PPS20} show
that---after extending the algorithm in a minor way---they satisfy an
even stronger property, which they call \emph{full justified
  representation}.  We reproduce this argument below, as it will be
helpful in understanding how the algorithm works. However, first we will
show that it always outputs a committee of size~$k$.

Let $I = (C,\calA, k)$ be an input instance for GreedyEJR. We note
that, whenever the algorithm chooses some $\ell$-large, $\ell$-cohesive set in the \emph{while} loop, it removes at least $\ell\cdot\frac{n}{k}$ voters
from $N$ and adds at most $\ell$ candidates to $W$, 
i.e., it removes at least $\frac{n}{k}$ voters for each candidate it adds. 
Thus, altogether,
the \emph{while} loop adds at most $k$ candidates to $W$.

\begin{proposition}[\citet{BFKN19,PPS20}]
\label{prop:greedy-EJR}
  Given an instance $I = (C, \calA, k)$, GreedyEJR outputs a committee
  $W$ such that $W \in \EJR(I)$.
\end{proposition}
\begin{proof}
  Let $I = (C, \calA, k)$ be an instance, and let $W$ be a committee
  that GreedyEJR outputs for this instance. For the sake of
  contradiction, let us assume that there is a number $\ell \in [k]$
  and an $\ell$-large, $\ell$-cohesive group of voters $F$
  witnessing that $W$ fails EJR (i.e., each voter in $F$ approves fewer than $\ell$ members of $W$).
  Let us consider the situation after GreedyEJR completed the
  \emph{while} loop for this value of $\ell$ (and all the larger
  ones). There are two possibilities:
  \begin{enumerate}
  \item At this point of the algorithm, some voter $i \in F$ was already
    removed from $N$. However, for this to happen GreedyEJR must have extended $W$
    with at least $\ell$ candidates approved by $i$, which contradicts our assumption.
  \item At this point, $N$ contains all members of $F$. Yet, this is
    impossible because $F$ is an $\ell$-large, $\ell$-cohesive group,
    so the \emph{for} loop would have to consider it and remove all
    its members from $N$.
  \end{enumerate}
  As we see, both of these possibilities lead to contradictions and,
  thus, the assumption that $F$ exists is false. This completes the
  proof.
\end{proof}

\begin{algorithm}[t]
  \DontPrintSemicolon
  \SetAlgoLined
  \KwIn{An instance $(C,\calA, k)$, where $A = (A_1, \ldots, A_n)$}
  \KwOut{A committee $W \subseteq C$ such that $|W| = k$ and $W$ provides EJR}
  \BlankLine
  $N \leftarrow \{1, \ldots, n\}$\;
  $W \leftarrow \varnothing$\;
  \For{$\ell \leftarrow k$ \KwTo  $1$}{
    \While{$N$ contains an $\ell$-cohesive, $\ell$-large group}{
        \emph{\small // choose a maximum-size $\ell$-cohesive, $\ell$-large group}\;
        $N' \in \arg\max_{M \subseteq N}\{|M| \colon |\bigcap_{i \in M}A_i| \geq \ell\}$\;
        $B \leftarrow$ a set of $\ell$ arbitrary candidates from $\bigcap_{i \in N'}A_i$\;
        $N \leftarrow N \setminus N'$\;
        $W \leftarrow W \cup B$\; \label{gejr:extend}
      }
    }
    \If{$|W| < k$}{
      extend $W$ to size $k$ arbitrarily 
    }
    \caption{\label{alg:greedy-ejr}GreedyEJR}
\end{algorithm}

We are now ready to prove Theorem~\ref{thm:ejr}.

\begin{proof}[Proof of Theorem~\ref{thm:ejr}]
  The lower bound follows from an example of \citet{LS20} (see Example~\ref{ex:coverage-EJR}), so we focus on the upper bound.
  Let us consider an instance $I = (C,\calA, k)$ with
  $\calA = (A_1, \ldots, A_n)$, and let $O$ be a size-$k$ committee
  that achieves the highest possible coverage for this instance, i.e., $\cov(O)=\cov(I)$.  We
  let $W$ be a committee that GreedyEJR would compute for instance
  $I$, assuming that in the final \emph{if} statement it chooses
  candidates so as to maximize the committee's coverage. 
  We will show that $\cov(W) \geq \frac{3}{4}\cdot \cov(O)$.

  Let $\alpha \in [0,1]$ be the fraction of first-stage candidates in
  $W$.  We claim that the first-stage candidates cover at least
  $\alpha n$ voters. 
  To see why this is the case,
  consider some iteration of the \emph{while} loop of
  GreedyEJR. Whenever an iteration is executed, GreedyEJR
  (a)~extends~$W$ with a set $B$ of $\ell$ candidates, and (b)~removes
  at least $\ell \cdot \frac{n}{k}$ voters from the set $N$ (i.e., the
  algorithm finds out that at least $\ell \cdot \frac{n}{k}$ more
  voters are covered).\footnote{Note that there is a subtlety
    here. Committee $W$ might have already contained some members
    of~$B$. In such a case, the voters removed from $N$ were already
    covered, but the algorithm did not ``realize'' this. Thus each
    execution of line~\ref{gejr:extend} means extending the committee
    with at most $\ell$ candidates and obtaining a proof that at least
    $\ell \cdot \frac{n}{k}$ more voters are covered.}  Thus, if
  GreedyEJR identifies $t$ first-stage candidates in total, then these
  candidates cover at least $t \cdot \frac{n}{k}$ voters.  Since
  $\alpha = \nicefrac{t}{k}$, we obtain our claim.

  Next let us consider the voters covered by committee $O$. We will
  refer to them as $N_O$; note that $|N_O| = \cov(I)$. Let $\beta$
  be the fraction of voters from $N_O$ that are covered by the
  first-stage candidates. As argued in the proof of Theorem~\ref{thm:jr-swcov}, 
  using a $1-\alpha$ fraction of the candidates from $O$ we can
  cover at least $1-\alpha$ fraction of the uncovered voters from
  $N_O$, i.e., altogether $(1-\alpha)(1-\beta)$ fraction of the
  voters from $N_O$. Now we consider two cases:
  \begin{enumerate}
  \item If $\beta \leq \alpha$ then we know that the second-stage
    candidates cover at least $(1-\beta)(1-\alpha) \geq (1-\alpha)^2$
    fraction of the voters from $N_O$. As the first-stage voters
    cover at least $\alpha n \geq \alpha |N_O|$ voters, it follows that altogether committee $W$ covers at least
    $(\alpha + (1-\alpha)^2)\cov(I)$ voters (recall that
    $\cov(I) = |N_O|$).
  \item If $\beta > \alpha$ then we know that the number of voters covered
    by the first-stage candidates is at least:
    \[ \beta |N_O| = \alpha |N_O| + (\beta-\alpha)|N_O|, \] whereas
    the number of voters covered by the second-stage candidates is at
    least:
    \begin{align*} 
    &(1-\beta)(1-\alpha)|N_O| \\
    &= (1-\alpha)^2|N_O| - (\beta -
      \alpha)(1-\alpha)|N_O|. 
    \end{align*} 
      If we add these two values, we obtain:
    \begin{align*} 
      &\alpha |N_O| + (1-\alpha)^2|N_O| \\
      &+ \underbrace{(\beta-\alpha)|N_O| - (\beta -
      \alpha)(1-\alpha)|N_O|}_{\geq 0}\\ &\geq (\alpha + (1-\alpha)^2)\cov(I),
    \end{align*}
    where we use the fact that $\cov(I) = |N_O|$.
  \end{enumerate}
  All in all, in either case our committee $W$ covers at least
  $(\alpha + (1-\alpha)^2)\cov(I)$ voters (although not all of them
  need to belong to $N_O$).  
  Since $\alpha + (1-\alpha)^2 \ge \nicefrac{3}{4}$ for all $\alpha\in[0,1]$, our proof is complete.
\end{proof}

\section{Optimizing Social Welfare under JR}\label{sec:opt}
In this section we focus on the complexity of maximizing the social welfare 
over the set of all committees that provide JR. 
We do not discuss maximizing coverage,
as this topic is already covered in the rich literature on the
approval-based Chamberlin--Courant rule; see, e.g., the works of
\citet{pro-ros-zoh:j:proportional-representation},
\citet{bou-lu:c:chamberlin-courant}, \citet{bet-sli-uhl:j:mon-cc}, 
\citet{sko-yu-fal-elk:j:sc-cc},
and \citet{GBSF21}.

\subsection{Hardness of Approximation}\label{sec:inapprox}
Interestingly, while there are polynomial-time algorithms for finding committees that provide JR \citep{ABC17}, and we can find a social welfare-maximizing committee by simply picking $k$ candidates with the highest number of approvals, finding a social welfare-maximizing JR committee turns out to be $\np$-hard. This was first observed by \citet{BFKN19};
we will now strengthen their hardness result to an inapproximability result.
We show that, 
for any constant $\varepsilon > 0$, it is $\np$-hard to
approximate $\sw_\JR(I)$ to within a factor of
$k^{1/2 - \varepsilon}$. This hardness result holds even when $n=2k$.

\begin{theorem} \label{thm:welfare-hardness} For any
  $\varepsilon \in (0, 1/2)$, the following problem is $\np$-hard:\footnote{When saying that this search problem is $\np$-hard, 
we mean that a polynomial-time algorithm for the problem
would enable us to solve all problems in $\np$ in polynomial time.}
  Given an instance $I = (C, \calA, k)$, find a committee
  $W \in \JR(I)$ such that
  $\sw(W) \geq \sw_\JR(I) / k^{1/2 - \varepsilon}$.
  This problem remains $\np$-hard even if $k = n/2$.
\end{theorem}

Since GreedyCC$^\sw$ runs in polynomial time and the optimal social welfare under JR is at most the optimal social welfare overall, Theorem~\ref{thm:overall-price-sw} implies that the inapproximability factor in Theorem~\ref{thm:welfare-hardness} is asymptotically tight.
To prove Theorem~\ref{thm:welfare-hardness}, we need some
additional background.

We say that two edges of an undirected graph are \emph{adjacent} if they share a vertex.
A set of edges $S$ forms an \emph{edge dominating set} if every
edge in the graph is adjacent to at least one edge from $S$.  We
write $\eds(G)$ to denote the size of a smallest edge dominating set of $G$.
In the {\sc Edge Dominating Set} problem we are given an undirected graph $G$, 
and the task is to compute $\eds(G)$.

In the next lemma we link the \textsc{Edge Dominating Set} problem
with finding a social welfare-maximizing JR committee.

\begin{lemma}\label{lem:red} There is a polynomial-time reduction
  that, given a graph $G = (V, E)$ and an integer $y \geq 2$ such that
  $|V| + y$ is even, produces an instance $I = (C, \calA, k)$ with
  $k = (|V|+y)/2$ and $n = 2k$ voters, where
  $\sw_\JR(I) = y k - (y - 2)\cdot\eds(G)$.
\end{lemma}

\begin{proof}
  Let $G = (V,E)$ be our input graph, and let $y\ge 2$ be an integer such
  that $|V|+y$ is even. We form an instance $I = (C, \calA, k)$ as
  follows.  We let $U = \{u_1, \dots, u_y\}$ be a set of $y$
  dummy voters, and let the set of voters be $V \cup U$. 
  We let $Z = \{z_1, \dots, z_{|V|}\}$ be a set of $|V|$ dummy candidates, and
  let the set of candidates be $C = E \cup Z$.
  The approval ballots are such that:
  \begin{enumerate}
  \item each vertex voter $v \in V$ approves exactly the edges she belongs to, and
  \item  each dummy voter in $U$ approves all the candidates in $Z$.
  \end{enumerate}
  We set $k = (|V|+y)/2$. Note that there are $n = 2k$ voters.  This
  completes the description of our reduction. It is clear that it is
  computable in polynomial time.  We will now show that
  $\sw_\JR(I) = y k - (y - 2) \cdot \eds(G)$.

  First, we will show that
  $\sw_\JR(I) \geq y k - (y - 2) \cdot \eds(G)$. Let $E' \subseteq E$ be
  an edge dominating set of $G$ of size $\eds(G)$.  Note that
  $|E'| \leq |V|/2$ and, thus, $|E'| < n/2 = k$.  Let
  $W = E' \cup \{z_1, \dots, z_{k - |E'|}\}$.  Note that $W$ provides
  JR. Indeed, we have $\nicefrac{n}{k} = 2$, and each cohesive
  group of voters either contains a dummy voter (who approves
  $z_1$) or consists of two vertices connected by an edge (in which
  case, by definition of the edge dominating set, at least one of the
  vertex-voters approves a candidate in $E'$). Further, we
  have:
  \begin{align*}
    \sw(W) 
    &= 2 |E'| + y \cdot (k - |E'|) 
    = y k - (y - 2)\eds(G).
  \end{align*}

  Next, we will show that $\sw_\JR(I) \leq y k - (y - 2) \cdot
  \eds(G)$. Consider any committee $W' \in \JR(I)$ of the form
  $W' = E' \cup Z'$, where $E' \subseteq E$ and $Z' \subseteq Z$. Note
  that $E'$ must be an edge dominating set of $G$; otherwise, if
  $e = \{v, v'\}$ were not covered, then voters $v$ and $v'$ would
  violate the JR condition.  This means that $|E'| \geq \eds(G)$.  It
  follows that:
\begin{align*}
  \sw(W') &= 2|E'| + y|Z'| = 2|E'| + y\cdot(k-|E'|) \\
         &= y k - (y - 2) |E'| \leq y k - (y - 2) \cdot \eds(G),
\end{align*}
which concludes our proof.
\end{proof}

To derive Theorem~\ref{thm:welfare-hardness}, we will also need the
following hardness result for EDS.

\begin{lemma} \label{lem:hardness-eds} For every constant
  $\zeta \in (0, \nicefrac{1}{2})$, the following problem is $\np$-hard: Given a
  graph $G = (V, E)$, distinguish between the following possibilities:
\begin{enumerate}
\item (YES) $\eds(G) \leq \frac{1}{2} (|V| - |V|^{1 - \zeta})$;
\item (NO) $\eds(G) \geq \frac{1}{2} (|V| - |V|^\zeta)$.
\end{enumerate}
\end{lemma}

The proof of Lemma~\ref{lem:hardness-eds} follows by combining a
reduction of~\citet{CC06} with a super-constant hardness of
approximation for the \textsc{Independent Set} problem. 
Recall that an \emph{independent set} of a graph is a set of vertices such that any pair of them are not connected by an edge; we use $\alpha(G)$ to denote the maximum size of an independent set of the graph~$G$. We will need the following
result by \citet{H96}.

\begin{lemma}[\citet{H96}] \label{lem:Hastad}
For any $\beta \in (0, 1/2)$, the following problem is $\np$-hard:
Given a graph $G' = (V', E')$, distinguish between
\begin{itemize}
\item (YES) $\alpha(G') \geq |V'|^{1 - \beta}$;
\item (NO) $\alpha(G') \leq |V'|^\beta$.
\end{itemize}
\end{lemma}

\begin{proof}[Proof of Lemma~\ref{lem:hardness-eds}]
Let $\beta = \zeta/2$. Following~\citet{CC06}, we construct a graph $G = (V, E)$ from the graph $G' = (V', E')$ in Lemma~\ref{lem:Hastad} as follows. 
Let $n'' = |V'| - \lceil |V'|^{1 - \beta} \rceil$, and let $V = V' \cup V''$ where $V'' = \{v''_1, \dots, v''_{n''}\}$. 
Furthermore, every vertex in $V''$ has an edge to all other vertices in $G$, whereas the graph induced on $V'$ is the same as $G'$. 
It is clear that the reduction runs in polynomial time.

(YES) Suppose that $\alpha(G') \geq |V'|^{1 - \beta}$. 
Let $S \subseteq G'$ be an independent set of size $\lceil |V'|^{1 - \beta} \rceil$ in $G'$, and let the vertices in $V' \setminus S$ be $v'_1, \dots, v'_{n''}$. 
Consider the set $E^* \subseteq E$ of edges $\{(v'_1, v''_1), \dots, (v'_{n''}, v''_{n''})\}$.
One can check that $E^*$ is an edge dominating set of $G$. 
As a result, we have 
\begin{align*}
\eds(G) 
&\leq |E^*| 
= n'' 
= |V'| - \lceil |V'|^{1 - \beta} \rceil  \\
&= \frac{1}{2}\left(|V| - \lceil |V'|^{1 - \beta} \rceil\right),
\end{align*}
which is at most $\frac{1}{2}\left(|V| - |V|^{1 - \zeta}\right)$ for any sufficiently large $|V'|$.

(NO) Suppose that $\alpha(G') \leq |V'|^{\beta}$.
Observe that $\alpha(G) = \alpha(G')$. 
Moreover, for any edge dominating set $E^*$, it is clear that the set of vertices with no adjacent edge in $E^*$ must form an independent set. 
This implies that
\begin{align*}
\eds(G) 
&\geq \frac{|V| - \alpha(G)}{2} \\
&= \frac{|V| - \alpha(G')}{2} \geq \frac{1}{2}\left(|V| - |V'|^{\beta}\right),
\end{align*}
which is at least $\frac{1}{2}\left(|V| - |V|^\zeta\right)$.
\end{proof}

With Lemmas~\ref{lem:red} and~\ref{lem:hardness-eds} in
hand, we are ready to prove Theorem~\ref{thm:welfare-hardness}.

\begin{proof}[Proof of Theorem~\ref{thm:welfare-hardness}]
  Let $\varepsilon \in (0, \nicefrac{1}{2})$ be an arbitrary
  constant, and let $\zeta = \varepsilon/2$. We consider the problem
  stated in Lemma~\ref{lem:hardness-eds} for this value of $\zeta$, with
  input graph $G = (V, E)$.
  We apply the reduction from Lemma~\ref{lem:red} with
  $y = \lceil \sqrt{|V|} \rceil$ or $y = \lceil \sqrt{|V|} \rceil + 1$
  (as needed for the parity condition) and obtain an instance
  $I = (C, \calA, k)$. This instance has $n = |V| + y$ voters and
  committee size $k = n/2$.  We will show that if we were able to
  approximate $\sw_\JR(I)$ to within a factor of
  $k^{1/2 - \varepsilon}$, then we would also be able to distinguish
  between the two cases from Lemma~\ref{lem:hardness-eds}.  Indeed, if
  the YES case holds, then we have:
  \begin{align*}
    \sw_\JR(I)
    &\textstyle = y k - (y - 2) \cdot \eds(G) \\
    &\textstyle \geq y\cdot \frac{|V| + y}{2} - (y - 2) \cdot \frac{1}{2} (|V| - |V|^{1 - \zeta}) \\
    &\textstyle \geq \frac{y |V|}{2} - y \cdot \frac{1}{2} (|V| - |V|^{1 - \zeta}) \\
    &\textstyle = \frac{1}{2} y \cdot |V|^{1 - \zeta} 
    \textstyle = \Theta(k^{3/2 - \zeta}),
  \end{align*}
  where the last equality follows from our choice of $y$.  On the
  other hand, in the NO case we have:
  \begin{align*}
    \sw_\JR(I)
    &\textstyle = y k - (y - 2) \cdot \eds(G) \\
    &\textstyle \leq y\cdot \frac{|V| + y}{2} - (y - 2) \cdot \frac{1}{2} (|V| - |V|^\zeta) \\
    &\textstyle = (|V| + y) + (y - 2) \cdot \frac{1}{2} \left(y + |V|^\zeta\right) \\
    &\textstyle \leq |V| + y + \frac{y}{2} (y + |V|^\zeta) \\
    &\textstyle = \frac{1}{2}\cdot y^2 + |V| + \frac{y}{2}\cdot |V|^\zeta +  y = \Theta(k),
  \end{align*}
  where the last equality follows because $y^2 = \Theta(|V|)$,
  $y \geq |V|^{\nicefrac{1}{4}} \geq |V|^\zeta$, and
  $k = \Theta(|V|)$.  Altogether, the gap between the two cases is
  more than $k^{1/2 - \varepsilon}$ for any sufficiently large $k$.
  If there were a polynomial-time algorithm that could find
  $W \in \JR(I)$ with
  $\sw(W) \geq \sw_\JR(I) / k^{1/2 - \varepsilon}$, then we could
  distinguish between the two cases in polynomial time.  However, by
  Lemma~\ref{lem:hardness-eds}, doing so is $\np$-hard.
\end{proof}

\subsection{The case $n=k$}
\label{sec:approx}

Our $k^{1/2-\varepsilon}$-inapproximability result 
holds for the case where
$n=2k$, and it is easy to extend it to some cases with smaller $k$.
For example, for $n = 4k$ it suffices to clone each voter. Yet, if $k$
is extremely small, e.g., if it is bounded by a constant, then, of
course, we can find a social welfare-maximizing committee in $\JR(I)$ in
polynomial time.  

Larger values of $k$ can also become
easier. In what follows, we will focus on the case $n = k$.
In this case, the JR axiom is equivalent to requiring perfect coverage:
every voter forms a group of size $\nicefrac{n}{k}$ and thus must be represented
in the committee (recall that we require that $A_i\neq\varnothing$ for all $i\in [n]$). It turns out that maximizing the social welfare
under this condition remains hard.

\begin{theorem}\label{thm:inapprox:n=k}
For some constant $\varepsilon > 0$, the following problem is $\np$-hard:
Given an instance $I = (C, \calA, k)$ with $n = k$, find $W \in \JR(I)$ with $\sw(W) \geq \sw_\JR(I) / (1 + \varepsilon)$.
\end{theorem}

We show the hardness of approximation via a reduction from {\sc Exact-3-Cover (X3C)}; this is a special case of the set cover problem where each set has size exactly $3$.
Specifically, we are given a universe $[u]$ and a collection $\calS = \{S_1, \dots, S_M\}$ of subsets of $[u]$, where each $S_i$ is of size $3$. 
The goal is to select as few subsets as possible that can cover the universe; we use $\opt_{\setcov}(u, \calS)$ to denote the optimum of a set cover instance $(u, \calS)$. 
The main property of our reduction is stated below.

\begin{lemma} \label{lem:red-k-equal-n}
There is a polynomial-time reduction that, given an {\sc X3C} instance $(u, \calS)$, produces an instance $I = (C, \calA, k)$ with $n = k$ such that $\sw_\JR(I) = 4(u + 4) - \opt_{\setcov}(u, \calS)$.
\end{lemma}

\begin{proof}
Let the set of voters of $I$ be $[u] \cup T$ where $T$ is a set of size $4$, and let the set of candidates $C$ be $\calS \cup Z$ where $Z$ is a set of size $u+4$.
Each voter $i \in [u]$ approves only the subsets it belongs to, whereas each voter in $T$ approves all the candidates in $Z$. 
This completes our reduction description; it clearly runs in polynomial time.

Note that each candidate in $\calS$ is approved by exactly three voters, while each candidate in $Z$ is approved by exactly four voters.
As a result, an optimal solution for $I$ is to pick an optimal set cover, and fill the remainder of the committee with candidates from $Z$.
Hence, we have
\begin{align*}
\sw_\JR(I) 
&= 3 \cdot \opt_{\setcov}(u, \calS) \\
&\qquad+ 4 (u + 4 - \opt_{\setcov}(u, \calS)) \\
&= 4(u + 4) - \opt_{\setcov}(u, \calS),
\end{align*}
as claimed.
\end{proof}

In addition to the reduction, we need the following APX-hardness of {\sc X3C} as a starting point (see, e.g., Lemma~5.4 of \citet{EFI21}).

\begin{lemma} \label{lem:x3c-hardness}
For some constant $\zeta \in (0, 1/3)$, the following problem is $\np$-hard: 
Given an {\sc X3C} instance $(u, \calS)$, distinguish between
\begin{itemize}
\item (YES) $\opt_{\setcov}(u, \calS) = u/3$;
\item (NO) $\opt_{\setcov}(u, \calS) \geq u(1/3 + \zeta)$.
\end{itemize}
\end{lemma}

We are now ready to prove Theorem~\ref{thm:inapprox:n=k}.
\begin{proof}[Proof of Theorem~\ref{thm:inapprox:n=k}]
Let $\zeta$ be the constant from Lemma~\ref{lem:x3c-hardness}, and let $\varepsilon = 0.1\zeta$. 
We will show that it is $\np$-hard to approximate $\sw_\JR(I)$ to within a factor of $1 + \varepsilon$ when $n = k$. 
Given any instance $(u, \calS)$ of X3C, we apply our reduction from Lemma~\ref{lem:red-k-equal-n} to produce an instance $I = (C, \calA, k)$ such that $n = k$. 
In the YES case, we have
\begin{align*}
\sw_\JR(I) 
&= 4(u + 4) - \opt_{\setcov}(u, \calS) 
\\&\geq 4u - u/3 = 11u/3,
\end{align*}
whereas in the NO case we have
\begin{align*}
\sw_\JR(I)
&= 4(u + 4) - \opt_{\setcov}(u, \calS)\\
&\leq 4(u + 4) - u(1/3 + \zeta)
= (11/3 - \zeta) u + 16.
\end{align*}
Thus, the multiplicative gap between the two cases is larger than $1 + \varepsilon$ for any sufficiently large $u$. 
If there were a polynomial-time algorithm that could find $W \in \JR(I)$ with $\sw(W) \geq \sw_\JR(I) / (1 + \varepsilon)$, then we would be able to distinguish between the two cases in polynomial time.
However, by Lemma~\ref{lem:x3c-hardness}, doing so is $\np$-hard.
\end{proof}

In spite of the hardness, the case $n=k$ admits a $2$-approximation algorithm, which means that it is easier than the case $n=2k$ (Theorem~\ref{thm:welfare-hardness}).

\begin{theorem}\label{thm:n=k}
There is a polynomial-time algorithm that, given an instance $I = (C, \calA, k)$ where $n = k$, finds a committee $W \in \JR(I)$ such that $\sw(W) \geq \sw_\JR(I) / 2$.
\end{theorem}

\begin{proof}
For every candidate $j \in C$, let $V_j \subseteq N$ denote the set of voters that approve $j$, i.e., $V_j = \{i \in N : j \in A_i\}$. 
Furthermore, we assume that there is an ordering $\sigma$ of the candidates which we can use to break ties.

Our algorithm works as follows:
\begin{enumerate}
\item Let $W = \varnothing$ and $V^{\uncov} = N$.
\item While $|W| < k$:
\begin{enumerate}
\item Let $j$ be the candidate that maximizes $|V_j \cap V^{\uncov}|$, where ties are broken by maximizing $|V_j|$ and then by~$\sigma$. \label{step:greedy-select}
\item $W \leftarrow W \cup \{j\}$ and $V^{\uncov} \leftarrow V^{\uncov} \setminus V_j$.
\end{enumerate}
\item Output $W$.
\end{enumerate}
It is clear that the output $W$ belongs to $\JR(I)$ and that the algorithm runs in polynomial time. 
The rest of the proof is devoted to showing the approximation guarantee.

To analyze the approximation ratio, let us divide the algorithm into three phases: (0) when the candidate $j$ found in Step~\ref{step:greedy-select} satisfies $|V_j \cap V^{\uncov}| \geq 2$; (1) when $j$ found in Step~\ref{step:greedy-select} satisfies $|V_j \cap V^{\uncov}| = 1$; and (2) when $j$ found in Step~\ref{step:greedy-select} satisfies $V_j \cap V^{\uncov} = \varnothing$. 
Let us write $W = W^0 \cup W^1 \cup W^2$, where $W^\ell$ denotes the set of candidates added in phase~$\ell$. 
Furthermore, let $V^1$ denote the set $V^{\uncov}$ at the beginning of phase~1 and, for every $i \in V^1$, let $c^i$ denote the candidate selected in Step~\ref{step:greedy-select} to cover $i$.

Next, consider an optimal solution $W^{\opt}$. 
For every voter $i \in V^1$, let $c^{i, \opt}$ denote a candidate in $W^{\opt}$ that $i$ approves (when there are multiple such candidates, we select the one with maximum $|V_c|$ and then break ties based on $\sigma$); 
at least one such candidate is guaranteed to exist since $W^{\opt} \in \JR(I)$. 
Notice that the candidates $c^{i, \opt}$ are distinct for all $i \in V^1$ since, by definition of $V^1$, these voters do not share any approved candidate. 
Observe also that 
\begin{align} \label{eq:opt-not-one}
W^{\opt} \setminus \{c^{i, \opt}\}_{i \in V^1} \subseteq (C \setminus W^1)
\end{align}
because, if $c^{i}$ is present in $W^{\opt}$ for some $i \in V^1$, we must have $c^{i, \opt} = c^i$ due to our tie-breaking rules. 

From the description of the algorithm (specifically, the tie-breaker according to $|V_j|$ in phase~2), when the candidates in $C \setminus W_1$ are sorted in nonincreasing order of $|V_j|$, the set $W^0 \cup W^2$ must contain $|W^2|$ top candidates. 
Among the elements in the multiset $\{|V_j|\}_{j\in C\setminus W^1}$, the average of the top $|W^2|$ elements is
no greater than
\begin{align} \label{eq:top-elements}
 \frac{\sw(W^0\cup W^2)}{|W^2|}.
\end{align}
On the other hand, from~\eqref{eq:opt-not-one}, the set $W^{\opt} \setminus \{c^{i, \opt}\}_{i \in V^1}$ is a subset of $C\setminus W^1$, and the average of $|V_j|$ for the candidates $j$ in this set is
\begin{align} \label{eq:average}
\frac{\sw(W^{\opt} \setminus \{c^{i, \opt}\}_{i \in V^1})}{|W^{\opt} \setminus \{c^{i, \opt}\}_{i \in V^1}|}.
\end{align}
Combining \eqref{eq:top-elements}, \eqref{eq:average}, and the fact that $|W^2| \le |W|-|W^1| = |W^{\opt} \setminus \{c^{i, \opt}\}_{i \in V^1}|$, we get
\begin{align} \label{eq:cross-multiply}
&|W^2| \cdot \sw(W^{\opt} \setminus \{c^{i, \opt}\}_{i \in V^1}) \nonumber \\
&\leq |W^{\opt} \setminus \{c^{i, \opt}\}_{i \in V^1}| \cdot \sw(W^0 \cup W^2).
\end{align}

Now, since each candidate in $W^0$ reduces the size of $V^{\uncov}$ by at least $2$, we have that $|W_2| \geq |W_0|$. 
Furthermore, observe that $|W^{\opt} \setminus \{c^{i, \opt}\}_{i \in V^1}| = n - |W^1| = |W_0| + |W_2| \leq 2|W_2|$. 
Plugging this into \eqref{eq:cross-multiply}, we have
\begin{align*}
\sw(W^{\opt} \setminus \{c^{i, \opt}\}_{i \in V^1}) \leq 2 \cdot \sw(W^0 \cup W^2).
\end{align*}
Recall that for every $i \in V^1$, $c^i$ is a candidate that maximizes $|V_c|$ among those in $A_i$. 
This implies that $\sw(c^{i, \opt}) \leq \sw(c^i)$. 
Summing this up over all $i \in V^1$, we get
\begin{align*}
\sw(\{c^{i, \opt}\}_{i \in V^1}) \leq \sw(W^1).
\end{align*}
Adding the above two inequalities, we arrive at
\begin{align*}
\sw(W^{\opt}) \leq 2 \cdot \sw(W^0 \cup W^2) + \sw(W^1) \leq 2 \cdot \sw(W).
\end{align*}
It follows that $\sw(W) \geq \sw_\JR(I)/2$, as desired.
\end{proof}

\subsection{GreedyCC Versus Optimal JR Social Welfare}\label{subsec:greedy-bad}
Many of our algorithms are based on GreedyCC. 
We will now show that, using this approach, we cannot break the $\Theta(\sqrt{k})$
barrier: GreedyCC may output a committee whose social welfare is a factor of
$\Theta(\sqrt{k})$ smaller than that of an optimal JR committee.  

\begin{proposition}\label{prop:greedy-opt-gap}
There exists an instance $I=(C, \calA, k)$ that admits a JR committee with social welfare $\Theta(k\sqrt{k})$, but on which GreedyCC may select a committee with social welfare $\Theta(k)$,
irrespective of which selection rule is used in the second phase.
\end{proposition}

\begin{proof}
  Let $t$ be some positive integer. We will construct an instance with $n = 16t^2$ voters.
  We start by forming a graph $G = (V,E)$ whose vertices are partitioned into two groups,
  $L$ and $R$, where $|L| = |R| = \nicefrac{n}{2}$. The vertices in
  $L$ form a clique, whereas the vertices in $R$ are not connected to
  each other. There are also $\nicefrac{n}{2}$ edges that form a
  perfect matching between $L$ and $R$.

  In our instance $I$, each vertex in $V$ is a voter, each
  edge $e=\{u, v\}$ in $E$ is a candidate (approved by $u$ and $v$), 
  and there is a set $S$ of $\nicefrac{n}{4}$
  additional candidates, all of whom are approved by the same (arbitrary)
  $\sqrt{n}$ voters from $R$. We set the committee size to be
  $k = \nicefrac{n}{2}$, so that a JR committee needs to represent
  all cohesive groups of size $2$.

  First, we note that there is a JR committee with social welfare
  $\Theta(n \sqrt{n})$. Indeed, it suffices to take the $\nicefrac{n}{4}$
  candidates in $S$ (who already ensure social welfare of
  $\Theta(n \sqrt{n})$), and $\nicefrac{n}{4}$ edge candidates that
  form a perfect matching among the vertices in $L$. To see why this
  committee provides JR, note that we have two types of cohesive groups of size at least $2$:
  (i) members of $R$ that approve the candidates in $S$ and (ii) pairs of voters 
  that correspond to an edge; for every group of the latter type, 
  at least one of its members is in $L$.

  Second, we note that with appropriate tie-breaking, GreedyCC selects
  a committee with social welfare of $\Theta(n)$. In the first step,
  GreedyCC must select one of the candidates in $S$ (which removes
  $\sqrt{n}$ vertices in $R$ from consideration). In the next
  $\nicefrac{n}{2} - \sqrt{n}$ steps, GreedyCC may select all the
  edges between the vertices in $L$ and the remaining vertices in $R$. Each such step
  removes one vertex from $R$ and one from $L$, so we are left with no
  vertices in $R$ and $\sqrt{n}$ vertices in $L$. At this point, the social welfare
  provided by the selected candidates is $\Theta(n)$.
  Irrespective of how GreedyCC selects the remaining $\sqrt{n}-1$ candidates, 
  the social welfare will not exceed $\Theta(n)$, because each candidate can contribute at most $\sqrt{n}$ to the social welfare.
\end{proof}

\subsection{Structured Preferences}
\label{sec:structured}
Another approach to circumvent intractability is to focus on instances where voters have structured preferences. In this section, we consider one such domain, namely, the 1D-VCR domain, recently introduced by \citet{GBSF21}. The name of the domain, 1D-VCR, stands for
1-dimensional voter/candidate range model.

\begin{definition}
  \label{def:1D-VCR}
  Given an instance $I=(C, \calA, k)$, with voter set $N$, we say that the voters have
  \emph{1D-VCR preferences} if there exists a collection
  $(x_a,r_a)_{a \in C \cup N}$ of points $x_a \in \R$ and nonnegative
  real values $r_a \in \R$ such that for each voter $i \in N$ and
  candidate $c \in C$, it holds that $i$ approves $c$ if and only if
  $|x_c -x_i| \le r_c +r_i$.
\end{definition}
Given an agent $a \in C \cup N$, we refer to $x_a$ and $r_a$ as 
the {\em position} and {\em radius} of $a$, respectively. We
assume that the positions and radii of the agents are specified in the input instance. This assumption carries no computational cost, since we can decide in polynomial time whether a given instance $I$ is a 1D-VCR instance and compute the respective positions and radii; 
this follows from the work of \citet{muller97} and \citet{rafiey}.

The 1D-VCR domain is a generalization of both the candidate interval (CI) and voter interval (VI) domains of \citet{EL15} (which
translate the classic definitions of
single-peaked~\citep{bla:b:polsci:committees-elections} and
single-crossing elections~\citep{mir:j:single-crossing,rob:j:tax} to
the approval setting):
the CI domain (resp., the VI domain) corresponds to all candidates 
(resp., voters) having zero radii.

Given a 1D-VCR instance $I=(C, \calA, k)$ and an agent $a \in C \cup N$, we
denote by $s(a)=x_a-r_a$ and $t(a)=x_a+r_a$ the leftmost and the
rightmost point of $a$'s range of influence, respectively.  We call
$[s(a),t(a)]$ the \emph{interval} of $a$.  Note that
voter $i$ approves candidate $c$ if and only if 
$i$'s interval $[s(i), t(i)]$ and
$c$'s interval $[s(c), t(c)]$
have a non-empty intersection, i.e.,
\begin{align}\label{eq2:1D}
s(c)  \le t(i) ~\mbox{and}~ s(i) \le t(c).
\end{align}

The main result of this section is that the problem of maximizing the
social welfare under the JR constraint is tractable for 1D-VCR
instances.
\begin{theorem}\label{thm:1D}
  There is a polynomial-time algorithm that, given a 1D-VCR instance
  $I=(C, \calA, k)$, computes a committee that maximizes the social
  welfare under JR.
\end{theorem}

We first prove Theorem~\ref{thm:1D} for the case where 
the intervals of the candidates have a non-nested
structure (Lemma~\ref{lem:representative}). 
Below, we denote by $\JR_{\le k}(I)$ the set of committees
of size at most $k$ that satisfy JR (with parameter $k$) for an instance
$I=(C, \calA, k)$.  Note that for this lemma, we do not make the usual
assumption that $|C| \ge k$.

\begin{lemma}\label{lem:representative}
  Let $I=(C, \calA, k)$ be a 1D-VCR instance such that for every pair
  of candidates $c,c' \in C$ with $t(c)\le t(c')$ it holds that $s(c) \le s(c')$.
  Then, there is a polynomial-time algorithm that,
  given $k^* \in [\min \{k,|C|\}]$, computes a committee
  $W \subseteq C$ that maximizes the social welfare under the
  constraints that $W \in \JR_{\leq k}(I)$ and $|W|=k^*$
  (or reports that no such committee exists).
\end{lemma}

\begin{proof}
Without loss of generality, we assume that the candidates $C=\{c_1,c_2,\ldots,c_{m}\}$ are sorted in such a way that $t(c_j) < t(c_h)$ implies $j<h$. 
By the condition of the lemma, we have $s(c_j) \le s(c_h)$ and $t(c_j) \le t(c_h)$ for all $j<h$, so Lemma~A.2 of \citet{EFI21} implies the following:

\begin{lemma}
\label{lem:intermediate}
Suppose that $j<\ell<h$. Then any voter who approves both $c_j$ and $c_h$ also approves $c_\ell$.
\end{lemma}

For each $j \in [m]$, we write $C_j=\{c_1,c_2,\ldots,c_j\}$; for each subset $W \subseteq C_{j}$, we say that $W$ \emph{satisfies JR up to $j$} if $W$ represents every group of voters $N'$ such that $|N'|\ge \nicefrac{n}{k}$ and there exists a candidate in $C_{j}$ that is approved by all voters in $N'$. 

\vspace{0.3cm}
\noindent
\emph{Description of the algorithm:} 
Intuitively, our algorithm proceeds from left to right. It keeps track of the last agent $c_j$ to be included in a desired committee of size $s$. If $c_j$ alone can 
already provide JR up to $j$, then the algorithm simply computes the maximum social welfare attained by choosing a subset of $C_j$ that includes $c_j$. Otherwise, the algorithm looks for a candidate $c_\ell$ to the left of $c_j$'s such that there is a committee of size $s-1$ in $C_j$ that together with $c_j$ provides JR up to $j$ and achieves optimal social welfare.

Formally, for each $j \in [m]$ and $s \in [\min \{k^*,j\}]$, we define $T[j, s]$ as follows: 
\begin{itemize}
    \item if $s < k^*$ then
    $T[j, s]$ is the maximum social welfare achievable by choosing a subset of $C_j$ of size $s$ that contains $c_j$ and provides JR up to $j$; it is $-\infty$ if no such subset exists; 
    \item if $s=k^*$ then $T[j, s]$ is the maximum social welfare achievable by choosing a subset of $C_j$ of size $s$ that contains $c_j$ and provides JR; it is $-\infty$ if no such subset exists.
\end{itemize}
Consider a committee $W \subseteq C$ of size $k^*$; it follows that $W$ maximizes the social welfare among committees of size $k^*$ that provide JR if and only if $\sw(W)=\max_{j \in [m]} T[j,k^*]$. Given the values $T[1, k^*], \dots, T[m, k^*]$, such a committee can be found by standard dynamic programming techniques (i.e., backtracking).

We will describe how to compute the quantities $T[j, s]$. 
For $j \in [m]$ and $s=1$, we set $T[j, s]= \sw(\{c_j\})$ if $s<k^*$ and $\{c_j\}$ provides JR up to $j$, or if $s=k^*$ and $\{c_j\}$ provides JR; otherwise we set $T[j, s]= - \infty$.

Next, suppose that $j \geq 2$ and $2 \le s < k^*$. We say that $\ell < j$ is \emph{permissible} if there is no index $h$ with the following properties: (i) $\ell <h <j$, and (ii) there is a group of at least $\nicefrac{n}{k}$ voters such that every voter in this group approves $c_{h}$ but none of them approves $c_{\ell}$ or $c_j$. We will compute the value of $T[j, s]$ as follows: 
\begin{itemize}
    \item If $\{c_j\}$ provides JR up to $j$, we set $T[j, s]$ to be the maximum social welfare achievable by choosing a subset of $C_j$ of size $s$ that includes $c_j$.
    \item Otherwise, we set $T[j, s] = 
    \max \{ T[\ell, s-1] + \sw(\{c_j\}) \mid ~\ell < j~\mbox{is permissible} \}$. 
\end{itemize}
Finally, if $j \geq 2$ and $s = k^*$, we set the value of $T[j, s]$ as follows:
\begin{itemize}
    \item If $\{c_j\}$ provides JR, we set $T[j, s]$ to be the maximum social welfare achievable by choosing a subset of $C_j$ of size $s$ that includes $c_j$. 
    \item If $\{c_j\}$ does not provide JR, but there is no group of voters $N'$, $|N'|\ge \nicefrac{n}{k}$, such that no voter in $N'$ approves $c_j$, while there exists a candidate $c_h$ with $h>j$ approved by every voter in $N'$, 
    we set $T[j, s] = 
    \max \{ T[\ell, s-1] + \sw(\{c_j\}) \mid ~\ell <j~\mbox{is permissible} \}$. 
    \item Otherwise we set $T[j, s]=-\infty$. 
\end{itemize}

For each entry $T[j,s]$, we also keep track of the committee that yields social welfare equal to the value of that entry.

\vspace{0.3cm}
\noindent
\emph{Correctness:} We will prove by induction that the algorithm correctly computes the desired value.

The claim is trivial for $j=1$.
Now, consider some $j\ge 2$, and assume that the claim holds up to $j-1$.
The claim clearly holds when $s < k^*$ and $\{c_j\}$ satisfies JR up to $j$, as well as when $s = k^*$ and $\{c_j\}$ satisfies JR. Thus, assume otherwise. 

\smallskip

\noindent \textbf{Case 1:} $s<k^*$. Let $W \subseteq C_j$ be a committee of size $s$ that achieves the highest social welfare under the constraints that $c_j \in W$ and $W$ satisfies JR up to $j$. 

We will first show that $T[j,s] \ge \sw(W)$. Let $c_{\ell}$ be the candidate in $W$ with maximum index $\ell < j$. To see that $\ell$ is permissible, suppose towards a contradiction that $B$ is a cohesive group of voters such that every member of $B$ approves a candidate $c_{h}$ with $\ell < h < j$, but none of them approves $c_{\ell}$ or $c_j$. Since $B$ is represented by $W$, there is a voter $i$ in $B$ who approves some candidate $c_t \in W$ with $p < \ell$. 
By Lemma~\ref{lem:intermediate}, $i$ also approves $c_{\ell}$, a contradiction. 

Now, we will show that $W'=W \cap C_{\ell}$ satisfies JR up to $\ell$. Suppose towards a contradiction that there exists a cohesive group $B'$ of voters who jointly approve a candidate $c_{h}$ with $1 \le h \le \ell$ but is not represented by $W'$. Recall that $s(c_{h}) \le s(c_{\ell})$. 
Since $B'$ is represented by $W$, there is a voter $i \in B'$ who approves candidate $c_{j}$. However, this voter $i$ then approves $c_{\ell}$ by Lemma~\ref{lem:intermediate}, a contradiction. Thus, we have 
\begin{align*}
T[j,s] &\ge T[\ell,s-1] + \sw(\{c_j\})\\
       &\ge \sw(W')+ \sw(\{c_j\}) = \sw (W),
\end{align*}
where $T[\ell,s-1] \ge \sw(W')$ holds by the induction hypothesis.

Conversely, to show that $T[j,s] \le \sw(W)$, let $c_{\ell}$ be a candidate such that $\ell$ is permissible and $T[\ell,s-1] \neq - \infty$. Then, there is a committee $W' \subseteq C_{\ell}$ of size $s-1$ that satisfies JR up to ${\ell}$, which by the definition of permissibility implies that $W' \cup \{c_j\}$ is a committee of size $s$ that satisfies JR up to $j$. Thus, $\sw(W' \cup \{c_j\}) \le \sw (W)$, which shows that $T[j,s] \le \sw(W)$. 

\smallskip

\noindent \textbf{Case 2:} $s=k^*$. Let $W \subseteq C_j$ be a committee of size $s$ that achieves the highest social welfare under the constraints that $c_j \in W$ and $W$ satisfies JR. 
To see that $T[j,s] \ge \sw(W)$, let $c_{\ell}$ be the candidate in $W$ with maximum index $\ell < j$. Similarly to Case~1, one can show that $\ell$ is permissible and $W'=W \cap C_{\ell}$ satisfies JR up to $\ell$. Now, it remains to show that there is no cohesive group of voters such that every voter of the group approves a candidate $c_{h}$ with $h>j$ but none of them approves $c_j$. Suppose towards a contradiction that there exists such a group $N'$. Then, since $N'$ is represented by $W$, there is a voter $i$ who approves some candidate $c_p \in W$ with $p< j$. But then Lemma~\ref{lem:intermediate} implies that $i$ approves $c_j$. Thus, we obtain a contradiction, and therefore we have 
\begin{align*}
T[j,s] &\ge T[\ell,s-1] + \sw(\{c_j\})\\
       &\ge \sw(W')+ \sw(\{c_j\}) = \sw (W),
\end{align*}
where $T[\ell,s-1] \ge \sw(W')$ holds by the induction hypothesis. 

Conversely, to show that $T[j,s] \le \sw(W)$, let $c_{\ell}$ be a candidate such that $\ell$ is permissible and $T[\ell,s-1] \neq - \infty$, and suppose that no cohesive group of voters jointly approve $c_h$ with $h>j$ while unanimously disapproving $c_j$. Then, there is a committee $W' \subseteq C_{\ell}$ of size $s-1$ that satisfies JR up to ${\ell}$. 
Similarly to Case~1, it can be verified that $W' \cup \{c_j\}$ is a committee of size $s$ that satisfies JR. Thus, $\sw(W' \cup \{c_j\}) \le \sw (W)$. This shows that $T[j,s] \le \sw(W)$.

\vspace{0.3cm}
\noindent
\emph{Running time:} The number of entries in the table $T$ is polynomial, and each entry can be computed in polynomial time.
In particular, checking whether each singleton committee $\{c_j\}$ satisfies JR or JR up to $j$ can be done in polynomial time by enumerating the family of maximum sets of voters who approve a common candidate. Similarly, checking whether each $\ell < j$ is permissible can be done in polynomial time.
\end{proof}

Our algorithm for the general case consists of two phases: First, we compute a committee (possibly of size smaller 
than $k$) that provides JR for the original instance
and only contains candidates whose intervals are maximal. 
Second, we supplement this committee with highly-approved candidates. 

To formalize this idea, we introduce the notion of \emph{representatives}. 
Given a 1D-VCR instance $I=(C, \calA, k)$, we say that candidate $c^*$ 
is a \emph{representative} of candidate $c$ if the 
interval of $c^*$ strictly includes that of $c$ and is a maximal interval with this property. We 
denote the set of all representatives by $C^*$; we have 
$C^*=\{\, c \in C \mid \nexists c' \in C: [s(c),t(c)] \subsetneq [s(c'),t(c')] \,\}$. 
We denote the instance restricted to the representatives by 
$I^*=(C^*, \calA|_{C^*}, k)$. 
By Lemma~A.1 of \citet{EFI21}, if voter 
$i$ approves $c$, then $i$ also approves its representative $c^*$. 
The following lemma states that a committee that consists of representatives satisfies JR if and only if it satisfies JR for the restricted instance.

\begin{lemma}\label{lem:JR-rep}
Let $I=(C, \calA, k)$ be a 1D-VCR instance, and let $I^*=(C^*, \calA|_{C^*}, k)$ be the instance 
restricted to the representatives. Then for each $W \subseteq C^*$ 
it holds that $W \in \JR_{\le k}(I^*)$ if and only if $W \in \JR_{\le k}(I)$.  
\end{lemma}

\begin{proof}
If a group of voters is cohesive in $I^*$, then it is also cohesive in $I$.
Therefore, $W\in\JR_{\le k}(I)$ implies $W\in\JR_{\le k}(I^*)$. 
Conversely, suppose that $W \in \JR_{\le k}(I^*)$. 
Lemma~A.1 of \citet{EFI21} implies that every group of voters of size at least $\nicefrac{n}{k}$ who jointly approve a candidate $c \in C$ 
also approve any of $c$'s representatives, and is therefore represented by $W$. 
Hence, $W \in \JR_{\le k}(I)$. 
\end{proof}

Lemma~\ref{lem:JR-rep} allows us to decompose our instance into its restrictions 
to $C^*$ and $C \setminus C^*$.

\begin{lemma}\label{lem:restricted}
Let $I=(C, \calA, k)$ be a 1D-VCR instance, and let $I^*=(C^*, \calA|_{C^*}, k)$ be 
the instance restricted to the representatives. Then, we have $\sw_{\JR}(I)=\alpha$, where 
\begin{align*}
\alpha = \max \{\, &\sw(W^*)+\sw(W') \mid W^* \in \JR_{\le k}(I^*),\\
&|W^*|+|W'|=k, W' \subseteq C \setminus C^* \,\}.
\end{align*}
\end{lemma}

\begin{proof}
We first argue that $\sw_{\JR}(I) \le \alpha$. 
Let $\mathcal{W}= \{W\subseteq C: |W|=k, \sw(W)=\sw_\JR(I)\}$, and pick a committee $W\in\mathcal W$ such that $|W\cap C^*|\ge |W'\cap C^*|$
for all $W'\in\mathcal W$.
Let $W^* =W \cap C^*$. Observe first that for each $c \in C \setminus C^*$, 
if $c \in W$, then some representative $c^*$ of $c$ also belongs to $W$. 
Indeed, if this is not the case. then replacing $c$ with one of its representatives $c^*$
preserves justified representation, weakly increases the social welfare,
and we have $|W\cap C^*|<((W\setminus\{c\})\cup\{c^*\})\cap C^*$, 
a contradiction with our choice of $W$.
Thus, every voter $i$ who approves a candidate $c \in W$ has an approved candidate $c^*$ in $W^*$. 
It follows that $W^*$ satisfies JR for $I$, and by Lemma~\ref{lem:JR-rep} also for $I^*$.
Hence, $\sw_{\JR}(I)$ can be written as follows: 
\[
\sw_{\JR}(I) = \sw(W) = \sw(W^*) + \sw(W \setminus W^*),
\]
where $W^* \in \JR_{\le k}(I^*)$ and $W \setminus W^* \subseteq C \setminus C^*$. 
This implies that $\sw_{\JR}(I) \le \alpha$. 

Conversely, to see that $\sw_{\JR}(I) \ge \alpha$, consider $W^* \in \JR_{\le k}(I^*)$ 
and $W' \subseteq C \setminus C^*$ such that $|W^*|+|W'|=k$ and $\sw(W^*)+\sw(W')=\alpha$. 
Let $W=W^* \cup W'$. Then, $W$ satisfies the JR requirement of the original instance $I$ by 
Lemma~\ref{lem:JR-rep}. Thus, $\sw_{\JR}(I) \ge \sw(W)$.
This implies that $\sw_{\JR}(I) \ge \alpha$. 
\end{proof}

We can now use Lemma~\ref{lem:restricted} in order 
to prove Theorem~\ref{thm:1D} by iterating over the size of $W^*$. 

\begin{proof}[Proof of Theorem \ref{thm:1D}]
Let $I=(C, \calA, k)$ be a 1D-VCR instance and let $I^*=(C^*, \calA|_{C^*}, k)$ 
be the instance restricted to the representatives. By Lemma~\ref{lem:restricted}, we can proceed as follows. We try all values of $k^*$ in $\{0,1,\dots,\min\{|C^*|,k\}\}$.
For each value of $k^*$, 
since the intervals of the candidates in $C^*$ are 
not properly contained in each other,  
we can use Lemma~\ref{lem:representative} 
to compute a committee $W^*$
that maximizes the social welfare subject to the requirements that 
$W^* \in \JR_{\le k}(I^*)$ and $|W^*|=k^*$ 
(or determine that no such committee exists).
We can then find a committee $W' \subseteq C \setminus C^*$ 
of size $k-k^*$ maximizing the social welfare among such committees, 
by choosing $k-k^*$ candidates with the highest number of approvals. 
For each value of $k^*$ for which $W^*$ exists, we evaluate
$\sw(W^*\cup W')$, and return the best of these committees.
\end{proof}

 \begin{figure*}
 \centering
      \begin{subfigure}{0.33\textwidth} \centering
     \includegraphics[scale =0.195]{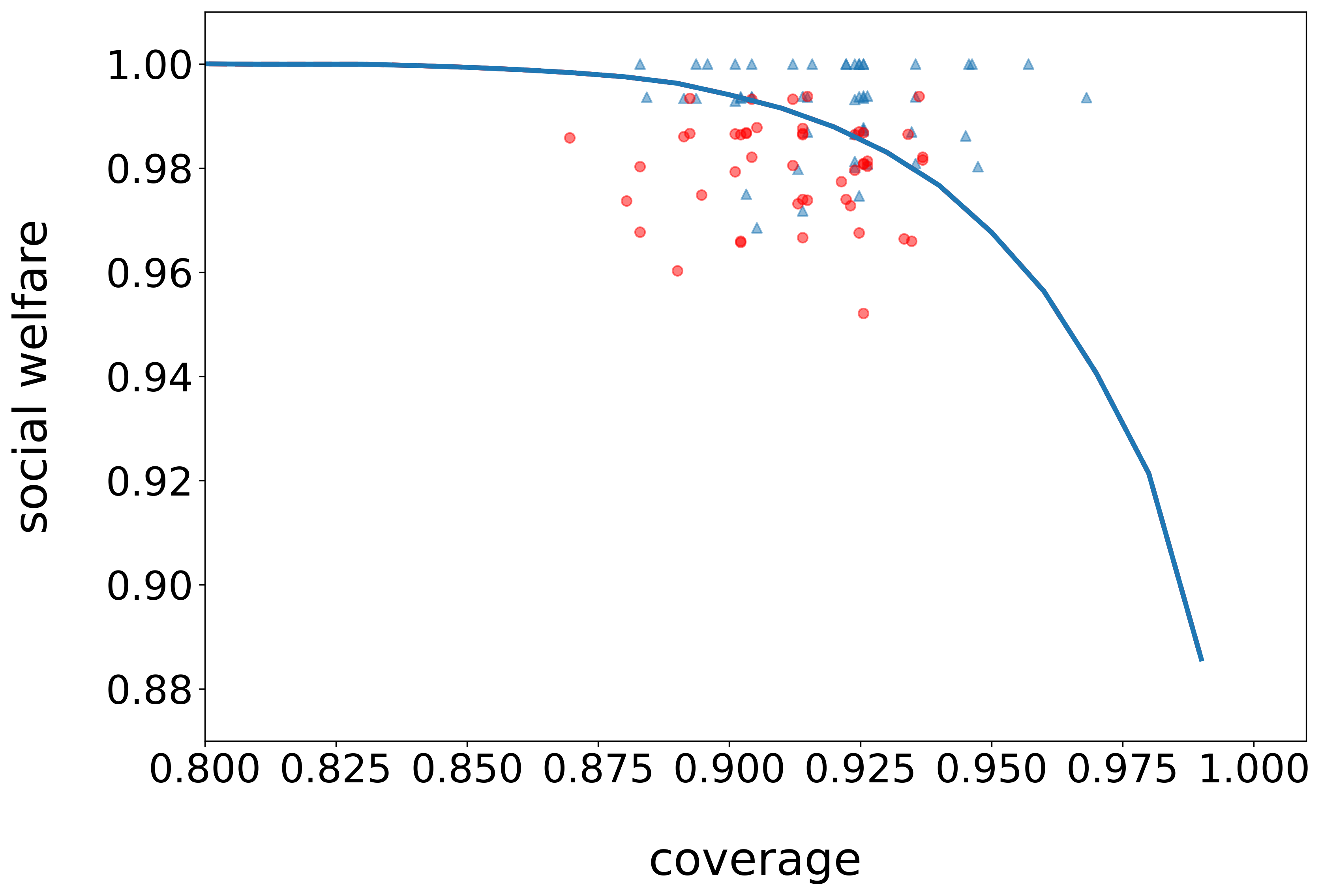}\caption{$0.1$-IC Model}\label{fig:topIC}
     \end{subfigure}
     \begin{subfigure}{0.33\textwidth}\centering
     \includegraphics[scale = 0.195]{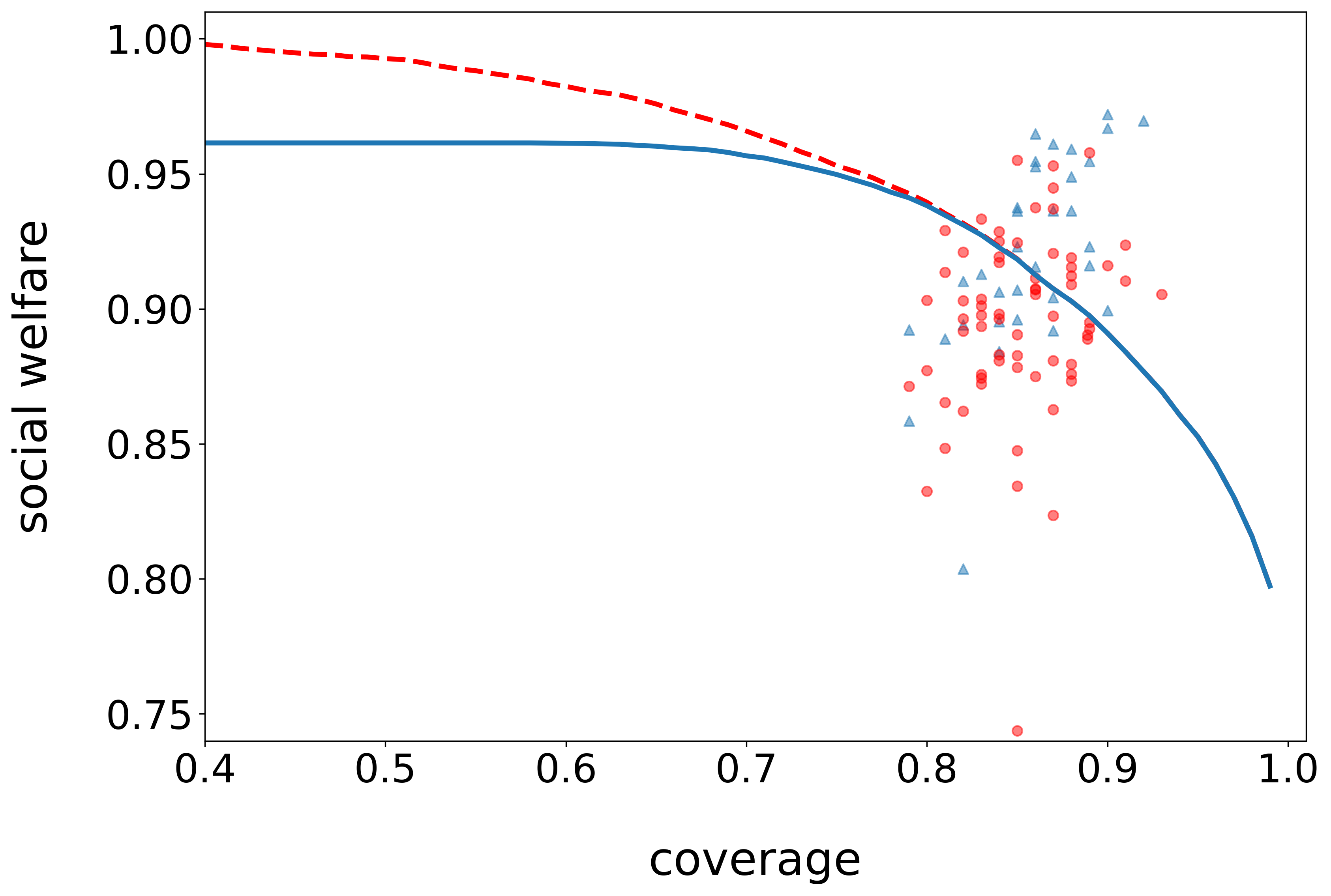}\caption{$0.054$-1D Euclidean Model}\label{fig:top1D}
     \end{subfigure}
     \begin{subfigure}{0.33\textwidth}\centering
     \includegraphics[scale=0.195]{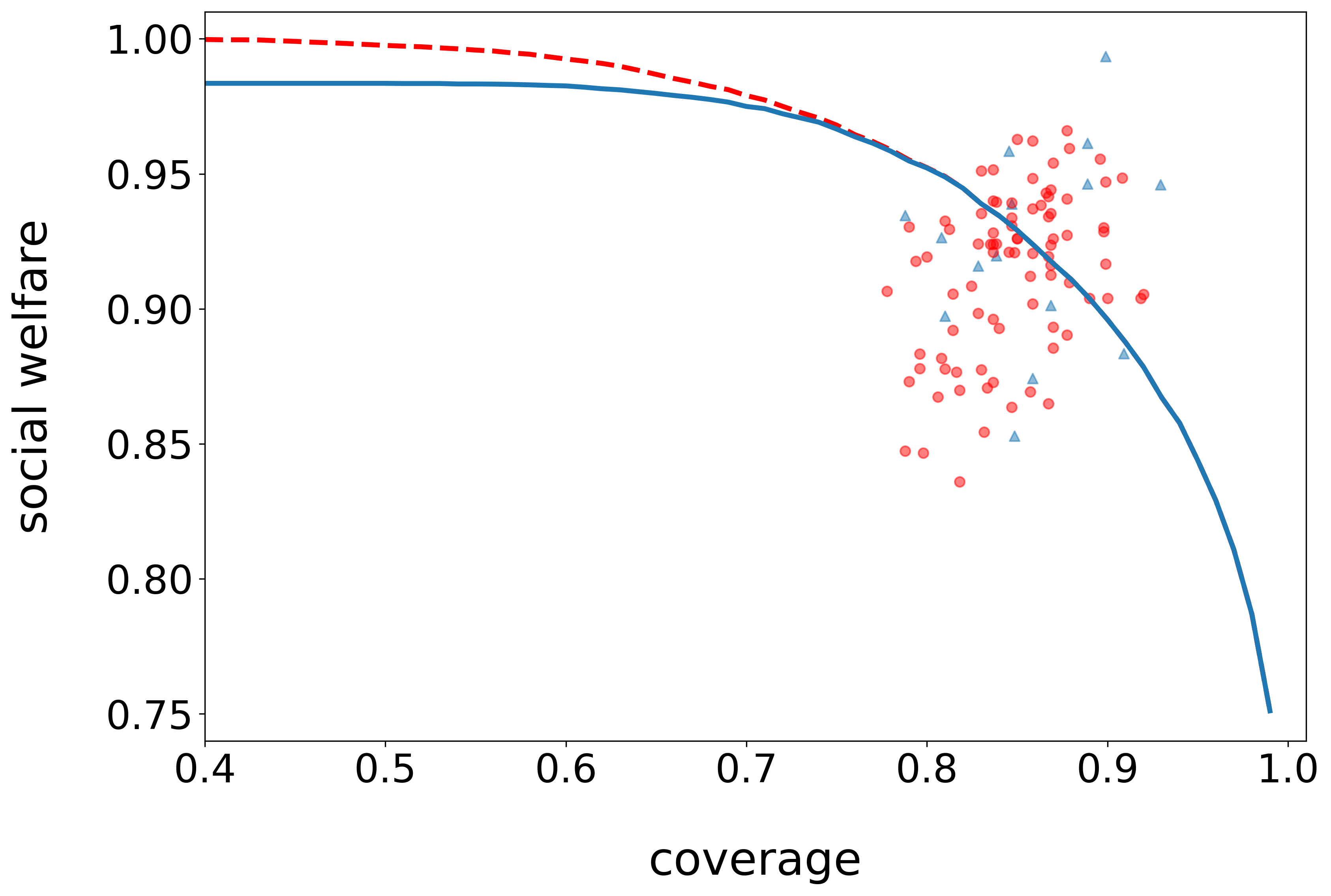}\caption{$0.195$-2D Euclidean Model}\label{fig:top2D}
     \end{subfigure}
     
    \begin{subfigure}{0.33\textwidth} \centering
     \includegraphics[scale =0.195]{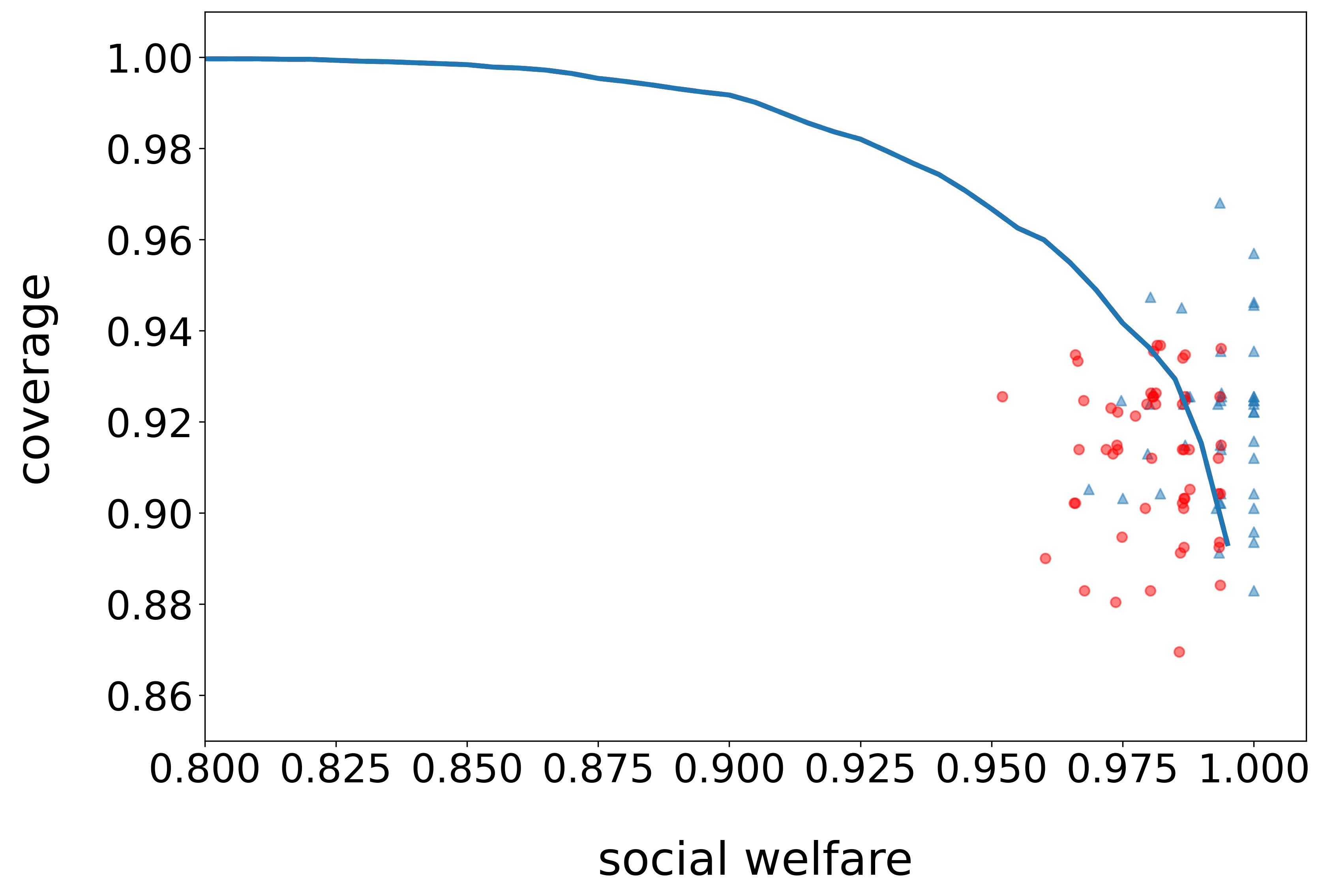}\caption{$0.1$-IC Model}\label{fig:bottomIC}
     \end{subfigure}
     \begin{subfigure}{0.33\textwidth}\centering
     \includegraphics[scale = 0.195]{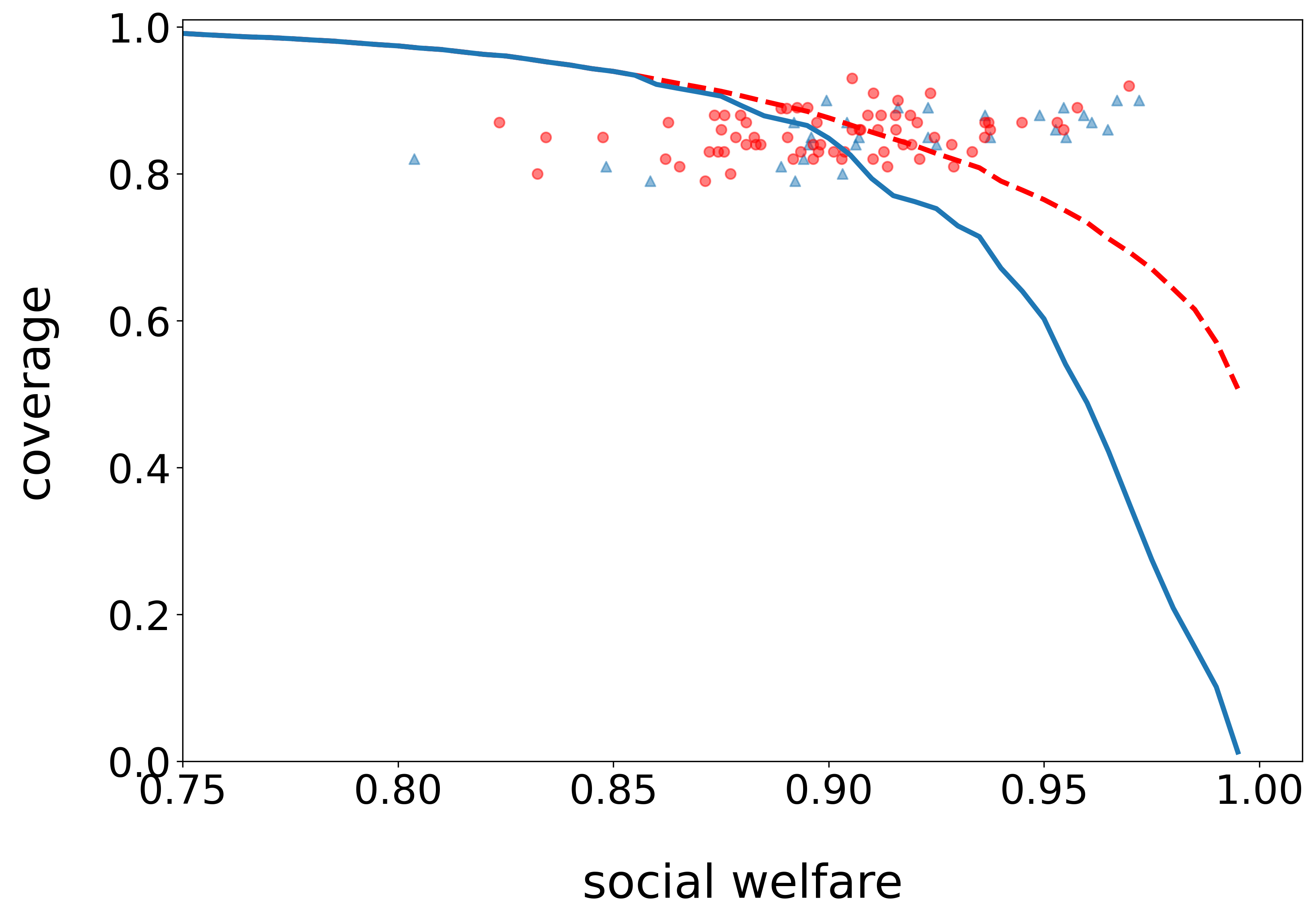}\caption{$0.054$-1D Euclidean Model}\label{fig:bottom1D}
     \end{subfigure}
     \begin{subfigure}{0.33\textwidth}\centering
     \includegraphics[scale=0.195]{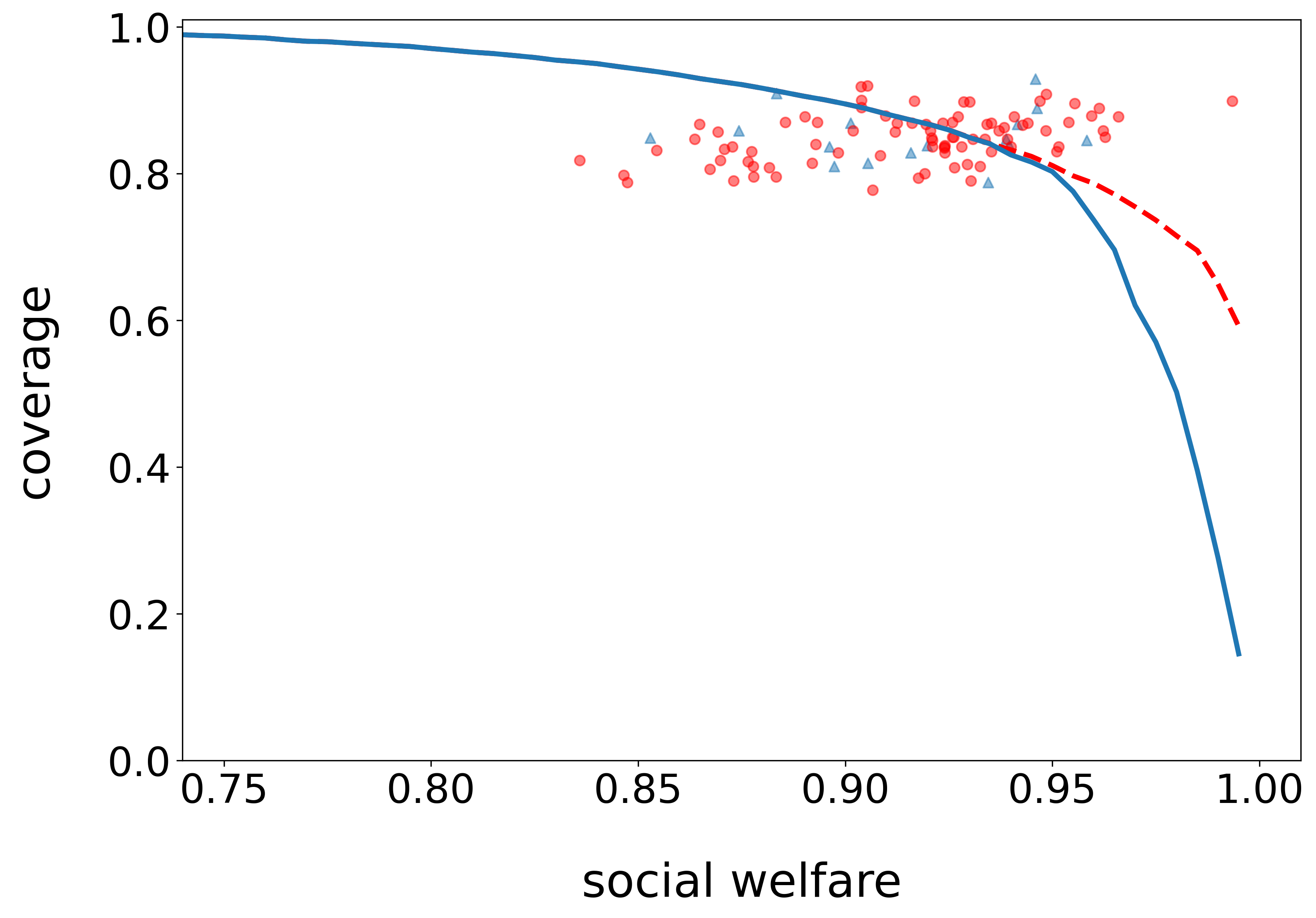}\caption{$0.195$-2D Euclidean Model}\label{fig:bottom2D}
     \end{subfigure}
     \caption{In the top row (resp., bottom row) we indicate the function $f$ (resp., $g$) by a blue solid line and the benchmark line (where we drop the JR requirement) by a red dashed line.  
     Dots and triangles indicate the committees computed by the greedy algorithm: 
     a committee is indicated by a blue triangle if it lies on the Pareto curve, and by a red dot otherwise.} \label{fig:pareto_curves1}
\end{figure*}

\section{Experimental Evaluation} \label{sec:experiment}

We analyze the trade-off between maximizing social welfare and maximizing coverage while requiring JR in randomly generated instances. 
We illustrate this trade-off in the form of Pareto curves, as depicted in Figure~\ref{fig:pareto_curves1}. In addition, we explore how well a variant of GreedyCC performs in selecting committees that are close to these Pareto curves. 

\paragraph{Setup} We consider three different models for generating random elections. All take as input the number of voters $n$, the number of candidates $m$, and one additional parameter. 
\begin{itemize}
\item The {\em impartial culture (IC)} model is parameterized by $p \in [0,1]$. Each candidate is approved by every voter independently with probability $p$. 
\item The {\em $n$-dimensional Euclidean} model takes as input a parameter $r \in [0,\sqrt{n}]$, also referred to as the \emph{radius}. Every voter $i$ and every candidate $c$ is associated with points $x_i, x_c \in [0,1]^n$, respectively, which are drawn uniformly at random. Candidate $c$ is approved by voter $i$ iff $d(x_i,x_c) \leq r$, where $d(\cdot,\cdot)$ is the Euclidean distance. We utilize the 1-dimensional (1D) and 2-dimensional (2D) Euclidean model in our experiments. 
\end{itemize} 

All three models are frequently used in the literature
(see, e.g., \citealp{BFKN19,GBSF21}).
We focus on elections with parameters $n=m=100$ and $k=10$. To make the results for the three models comparable, we chose the parameters $p$ and $r$ so that on average each voter approves $\nicefrac{n}{k}=10$ candidates. \citet{BFKN19} observe that, for the above models, JR is especially demanding for this ratio. This leads to parameters $p=0.1$ (for IC), $r=0.054$ (for 1D) and $r = 0.195$ (for 2D). 

Part of our implementation builds upon code contributed by Andrzej Kaczmarczyk \citep{BFKN19}. Our code can be found at https://github.com/Project-PRAGMA/PriceOfJR-AAAI-2022.

\paragraph{Trade-off between social welfare and coverage under JR} The Pareto curves depicted in the top row of Figure~\ref{fig:pareto_curves1} are created as follows. For each of the three models we sample $100$ instances. For a given instance $I$, we iterate over $\alpha \in [0,1]$ (in $0.01$ increments) and compute a size-$k$ committee $W_{I,\alpha}$ so as to maximize the social welfare subject to the constraint that at least $\alpha \cdot \cov(I)$ coverage is achieved and JR is satisfied. 
To compute $W_{I, \alpha}$, we use an IP formulation. 
Then, for each $\alpha$ we define $f(\alpha)$ as the average of $\sw(W_{I,\alpha})/\sw(I)$ over all generated  instances of the model and illustrate $f$ by a blue line in \Cref{fig:topIC,fig:top1D,fig:top2D}. For comparison, we repeat this process without requiring committees to provide JR, and depict the resulting curve by a red dashed line. 

The bottom row of Figure~\ref{fig:pareto_curves1} is created by exchanging the roles of the two objectives. That is, for $\alpha \in [0,1]$ (with $0.005$ increments) the committee $W_{I,\alpha}$ is chosen so as to maximize coverage subject to the constraints that the social welfare is at least $\alpha \cdot \sw(I)$ and JR is satisfied. There is one important difference to the previous experiments, namely, that for large values of $\alpha$, such a committee need not exist since high social welfare and JR can be incompatible. In these cases, we set $W_{I,\alpha}=\varnothing$. Then, we compute $g(\alpha)$ as the average of $\cov(W_{I,\alpha})/\cov(I)$ over $100$ instances, and illustrate it by a blue line in the bottom row of  Figure~\ref{fig:pareto_curves1}. Again, for comparison we drop the JR constraint and repeat the computation; the result is depicted by a red dashed line. Observe that the two figures in each column are not symmetric.

\paragraph{Greedy Heuristic}
In order to find committees that are close to the Pareto curves, is it necessary to solve an IP? To answer this question, we implement a greedy heuristic and analyze the distance from its selected committees to the Pareto curves. This heuristic  proceeds as follows. 

First, the algorithm computes $\sw(I)$ and a greedy estimate of $\cov(I)$, which is done by computing a committee $\overline{W}$ returned by GreedyCC. Then, the algorithm starts with an empty committee $W=\varnothing$, and in each iteration decides between one of two steps: In a \emph{coverage step}, the algorithm selects the next candidate $c$ so as to maximize $\cov(W \cup \{c\})$; among all such candidates, it chooses one with maximum approval score. In a \emph{social welfare step}, the algorithm selects a candidate $c$ with maximum approval score; among all such candidates $c$ it chooses one that maximizes $\cov(W \cup \{c\})$. 
As long as JR is not satisfied, the greedy algorithm performs coverage steps. After that, it compares $\cov(W)/\cov(\overline{W})$ with $\sw(W)/\sw(I)$. If the former is smaller, it performs a coverage step; otherwise it performs a social welfare step. 

After termination, we compute the exact approximation ratio of $W$, i.e., the point $(\cov(W)/\cov(I),\sw(W)/\sw(I))$ and indicate it by a blue triangle or a red dot in \Cref{fig:topIC,fig:top1D,fig:top2D}. More precisely, we depict the point by a red dot if it is not on the Pareto curve of the instance $I$. That is, we check whether there exists a committee $W'$ achieving coverage at least $\cov(W)$ and social welfare strictly greater than $\sw(W)$. In this case we also note the distance $(\sw(W') - \sw(W))/\sw(I)$ for the optimal $W'$. Otherwise, the point lies on the Pareto curve and is indicated by a blue triangle. For the bottom row of Figure~\ref{fig:pareto_curves1} we perform the analogous procedure.

\paragraph{Maximizing social welfare with respect to $\alpha$-coverage}

Consider the top row of Figure~\ref{fig:pareto_curves1}. For the impartial culture model with approval probability $p=0.1$ (\Cref{fig:topIC}), the restriction to JR committees has no effect as the two Pareto curves coincide. Also, the trade-off between social welfare and coverage does not appear to be very strong. Even when demanding optimal coverage, we can obtain (on average) $88\%$ of optimal social welfare. Regarding the performance of the greedy algorithm, $52\%$ of the committees are positioned on the Pareto curve of their instance. Among the remaining points, the average distance towards their respective Pareto curve (along the $\sw$-axis) is $0.012$. \footnote{We observe a grid-like pattern of the greedy points along both axes. See \Cref{app:experiments} for an explanation of this effect.}

The trade-off between the two objectives becomes more apparent for the Euclidean models (\Cref{fig:top1D,fig:top2D}). Even without any constraints regarding the coverage of the committee, the constraint that the committee has to provide JR induces a social welfare loss of roughly $4\%$ (for the 1D model) and $1.5\%$ (for the 2D model). Although the absolute loss, as compared to the best achievable social welfare, becomes larger for higher values of $\alpha$, namely $21.5\%$ for the 1D model and $25\%$ for the 2D model, the distance from the benchmark line diminishes. 
Regarding the performance of the greedy heuristic, $33\%$ (for 1D) and $15\%$ (for 2D) of the committees are positioned on the Pareto curve of their instance. The average distance of the remaining points from their respective Pareto curve is $0.022$ (for 1D) and $0.024$ (for 2D).

\paragraph{Maximizing coverage with respect to $\alpha$-social welfare} \Cref{fig:bottomIC} shows that, again, for the IC model, the JR constraint has no impact on maximizing coverage under the constraint that an $\alpha$-fraction of social welfare is achieved. Regarding the performance of the greedy heuristic, $45\%$ of the committees are positioned on the Pareto curve of their instance. For the remaining points, the distance towards the Pareto curve (along the $\cov$-axis) is $0.026$ on average. 
For the 1D model (\Cref{fig:bottom1D}), requiring JR has no impact on the maximum coverage achievable up to a social welfare fraction of $85\%$. Starting from there, more and more instances do not admit a committee that provides JR while attaining sufficient social welfare.
For the 2D model (\Cref{fig:bottom2D}) the situation is similar, though JR becomes restrictive only for $\alpha \geq 0.94$. Considering the greedy heuristic, $29\%$ (for 1D) and $16\%$ (for 2D) of the committees are positioned on the Pareto curves of their instance. For the remaining points, the average distance from their Pareto curve (along the $\cov$-axis) is $0.036$ and $0.038$ for the 1D and 2D model, respectively.

\paragraph{Conclusions of experiments} In contrast to our theoretical results, the experiments show that in practice, providing JR is often compatible with high social welfare. Moreover, for all three models, if we require at least $80\%$ coverage, demanding JR on top does not lead to additional social welfare loss. 
Our greedy heuristic performs well in finding committees that are close to the Pareto curve: For all models, more than $20\%$ of the committees are optimal and the remaining ones are close to the Pareto curve on average. 

\Cref{app:experiments} contains a comparison to experiments from the literature and details on the computational infrastructure.

\section{Conclusions and Future Work}
\label{sec:conclusion}

In this paper, we have provided an extensive analysis of the impact of the justified representation (JR) axiom on both the (utilitarian) social welfare and the coverage, both from a worst-case perspective and from an algorithmic perspective, and complemented our theoretical results with an empirical analysis. 
While we have shown that the price of JR is rather high in the worst case, our experimental results demonstrate that it is often low in practice---this implies that in many cases JR can be achieved in conjunction with traditional desiderata.
In addition, although maximizing the social welfare under JR is hard even to approximate, an exact solution can be found efficiently when the preferences admit a one-dimensional representation.

For future work, it would be interesting to perform a similar analysis for the EJR axiom, as well as other proportionality axioms such as PJR \citep{FEL17} and the core \citep{ABC17}; Theorem~\ref{thm:ejr} is the first step in this direction.
Furthermore, strengthening \Cref{thm:overall-price-sw} so that it matches the $\nicefrac{\sqrt{k}}{2}$ lower bound for all committee sizes $k$ would be useful as well.

\section*{Acknowledgments}

This project has received funding from the European 
    Research Council (ERC) under the European Union’s Horizon 2020 
    research and innovation programme (grant agreement No 101002854), 
    from the Deutsche Forschungsgemeinschaft under grant
BR 4744/2-1, from JST PRESTO under grant number JPMJPR20C1, and from an NUS Start-up Grant.
We would like to thank the anonymous reviewers for their valuable comments.
\vspace{0.5mm}
\begin{center}
    \noindent \includegraphics[width=2.5cm]{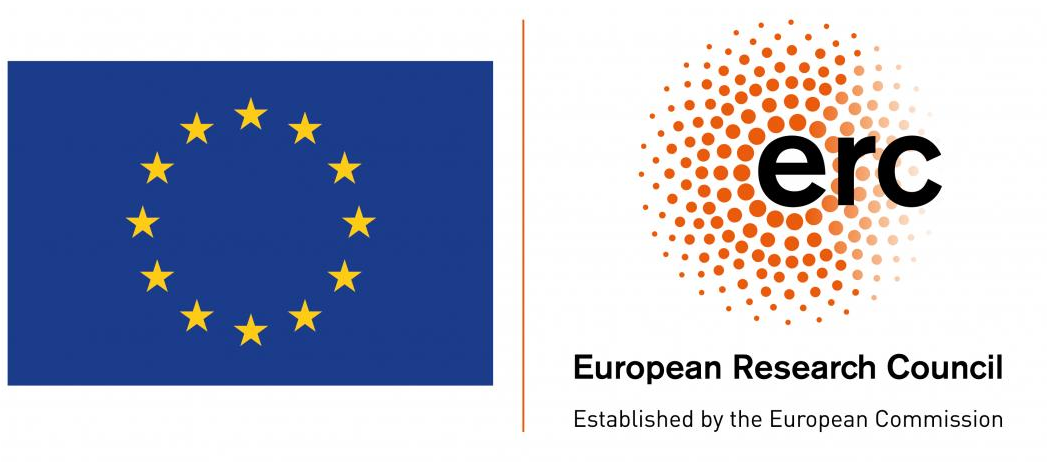}
\end{center}

\bibliography{main}

\begin{thebibliography}{26}
\providecommand{\natexlab}[1]{#1}

\bibitem[{Aziz et~al.(2017)Aziz, Brill, Conitzer, Elkind, Freeman, and
  Walsh}]{ABC17}
Aziz, H.; Brill, M.; Conitzer, V.; Elkind, E.; Freeman, R.; and Walsh, T. 2017.
\newblock Justified representation in approval-based committee voting.
\newblock \emph{Social Choice and Welfare}, 48(2): 461--485.

\bibitem[{Betzler, Slinko, and Uhlmann(2013)}]{bet-sli-uhl:j:mon-cc}
Betzler, N.; Slinko, A.; and Uhlmann, J. 2013.
\newblock On the computation of fully proportional representation.
\newblock \emph{Journal of Artificial Intelligence Research}, 47: 475--519.

\bibitem[{Black(1958)}]{bla:b:polsci:committees-elections}
Black, D. 1958.
\newblock \emph{The Theory of Committees and Elections}.
\newblock Cambridge University Press.

\bibitem[{Bredereck et~al.(2019)Bredereck, Faliszewski, Kaczmarczyk, and
  Niedermeier}]{BFKN19}
Bredereck, R.; Faliszewski, P.; Kaczmarczyk, A.; and Niedermeier, R. 2019.
\newblock An experimental view on committees providing justified
  representation.
\newblock In \emph{Proceedings of the 28th International Joint Conference on
  Artificial Intelligence (IJCAI)}, 109--115.

\bibitem[{Chamberlin and Courant(1983)}]{CC83}
Chamberlin, J.~R.; and Courant, P.~N. 1983.
\newblock Representative deliberations and representative decisions:
  Proportional representation and the {B}orda rule.
\newblock \emph{American Political Science Review}, 77(3): 718--733.

\bibitem[{Chleb{\'{\i}}k and Chleb{\'{\i}}kov{\'{a}}(2006)}]{CC06}
Chleb{\'{\i}}k, M.; and Chleb{\'{\i}}kov{\'{a}}, J. 2006.
\newblock Approximation hardness of edge dominating set problems.
\newblock \emph{Journal of Combinatorial Optimization}, 11(3): 279--290.

\bibitem[{Elkind et~al.(2021)Elkind, Faliszewski, Igarashi, Manurangsi,
  Schmidt-Kraepelin, and Suksompong}]{EFI21}
Elkind, E.; Faliszewski, P.; Igarashi, A.; Manurangsi, P.; Schmidt-Kraepelin,
  U.; and Suksompong, W. 2021.
\newblock Justifying groups in multiwinner approval voting.
\newblock \emph{arXiv preprint arXiv:2108.12949}.

\bibitem[{Elkind and Lackner(2015)}]{EL15}
Elkind, E.; and Lackner, M. 2015.
\newblock Structure in dichotomous preferences.
\newblock In \emph{Proceedings of the 24th International Joint Conference on
  Artificial Intelligence (IJCAI)}, 2019--2025.

\bibitem[{Faliszewski et~al.(2017)Faliszewski, Skowron, Slinko, and
  Talmon}]{FSST17}
Faliszewski, P.; Skowron, P.; Slinko, A.; and Talmon, N. 2017.
\newblock Multiwinner voting: a new challenge for social choice theory.
\newblock In Endriss, U., ed., \emph{Trends in Computational Social Choice},
  chapter~2, 27--47. AI Access.

\bibitem[{Fern{\'{a}}ndez et~al.(2017)Fern{\'{a}}ndez, Elkind, Lackner,
  Garc{\'{\i}}a, Arias{-}Fisteus, Basanta{-}Val, and Skowron}]{FEL17}
Fern{\'{a}}ndez, L.~S.; Elkind, E.; Lackner, M.; Garc{\'{\i}}a, N.~F.;
  Arias{-}Fisteus, J.; Basanta{-}Val, P.; and Skowron, P. 2017.
\newblock Proportional justified representation.
\newblock In \emph{Proceedings of the 31st {AAAI} Conference on Artificial
  Intelligence (AAAI)}, 670--676.

\bibitem[{Godziszewski et~al.(2021)Godziszewski, Batko, Skowron, and
  Faliszewski}]{GBSF21}
Godziszewski, M.; Batko, P.; Skowron, P.; and Faliszewski, P. 2021.
\newblock An analysis of approval-based committee rules for {2D}-{E}uclidean
  elections.
\newblock In \emph{Proceedings of the 35th {AAAI} Conference on Artificial
  Intelligence (AAAI)}, 5448--5455.

\bibitem[{H{\aa}stad(1996)}]{H96}
H{\aa}stad, J. 1996.
\newblock Clique is hard to approximate within $n^{1 - \varepsilon}$.
\newblock In \emph{Proceedings of the 37th Annual Symposium on Foundations of
  Computer Science (FOCS)}, 627--636.

\bibitem[{Kocot et~al.(2019)Kocot, Kolonko, Elkind, Faliszewski, and
  Talmon}]{KKE+19a}
Kocot, M.; Kolonko, A.; Elkind, E.; Faliszewski, P.; and Talmon, N. 2019.
\newblock Multigoal committee selection.
\newblock In \emph{Proceedings of the 28th International Joint Conference on
  Artificial Intelligence (IJCAI)}, 385--391.

\bibitem[{Krause and Golovin(2014)}]{kra-gol:b:submodular}
Krause, A.; and Golovin, D. 2014.
\newblock Submodular function maximization.
\newblock In \emph{Tractability: Practical Approaches to Hard Problems},
  chapter~3, 71--104. Cambridge University Press.

\bibitem[{Lackner and Skowron(2020{\natexlab{a}})}]{LS-survey}
Lackner, M.; and Skowron, P. 2020{\natexlab{a}}.
\newblock Approval-based committee voting: Axioms, algorithms, and
  applications.
\newblock \emph{arXiv preprint arXiv:2007.01795}.

\bibitem[{Lackner and Skowron(2020{\natexlab{b}})}]{LS20}
Lackner, M.; and Skowron, P. 2020{\natexlab{b}}.
\newblock Utilitarian welfare and representation guarantees of approval-based
  multiwinner rules.
\newblock \emph{Artificial Intelligence}, 288: 103366.

\bibitem[{Lu and Boutilier(2011)}]{bou-lu:c:chamberlin-courant}
Lu, T.; and Boutilier, C. 2011.
\newblock Budgeted social choice: From consensus to personalized decision
  making.
\newblock In \emph{Proceedings of the 22nd International Joint Conference on
  Artificial Intelligence (IJCAI)}, 280--286.

\bibitem[{Mirrlees(1971)}]{mir:j:single-crossing}
Mirrlees, J. 1971.
\newblock An exploration in the theory of optimal income taxation.
\newblock \emph{Review of Economic Studies}, 38(2): 175--208.

\bibitem[{M{\"u}ller(1997)}]{muller97}
M{\"u}ller, H. 1997.
\newblock Recognizing interval digraphs and interval bigraphs in polynomial
  time.
\newblock \emph{Discrete Applied Mathematics}, 78(1--3): 189--205.

\bibitem[{Nemhauser, Wolsey, and Fisher(1978)}]{nem-wol-fis:j:submodular}
Nemhauser, G.; Wolsey, L.; and Fisher, M. 1978.
\newblock An analysis of approximations for maximizing submodular set
  functions.
\newblock \emph{Mathematical Programming}, 14(1): 265--294.

\bibitem[{Peters, Pierczynski, and Skowron(2021)}]{PPS20}
Peters, D.; Pierczynski, G.; and Skowron, P. 2021.
\newblock Proportional participatory budgeting with additive utilities.
\newblock In \emph{Proceedings of the 35th Conference on Neural Information
  Processing Systems (NeurIPS)}.
\newblock Forthcoming.

\bibitem[{Procaccia, Rosenschein, and
  Zohar(2008)}]{pro-ros-zoh:j:proportional-representation}
Procaccia, A.; Rosenschein, J.; and Zohar, A. 2008.
\newblock On the complexity of achieving proportional representation.
\newblock \emph{Social Choice and Welfare}, 30(3): 353--362.

\bibitem[{Rafiey(2012)}]{rafiey}
Rafiey, A. 2012.
\newblock Recognizing interval bigraphs by forbidden patterns.
\newblock \emph{arXiv preprint arXiv:1211.2662}.

\bibitem[{Roberts(1977)}]{rob:j:tax}
Roberts, K.~W.~S. 1977.
\newblock Voting over income tax schedules.
\newblock \emph{Journal of Public Economics}, 8(3): 329--340.

\bibitem[{Skowron and Faliszewski(2017)}]{SF17}
Skowron, P.; and Faliszewski, P. 2017.
\newblock Chamberlin--{C}ourant rule with approval ballots: Approximating the
  {MaxCover} problem with bounded frequencies in {FPT} time.
\newblock \emph{Journal of Artificial Intelligence Research}, 60: 687--716.

\bibitem[{Skowron et~al.(2015)Skowron, Yu, Faliszewski, and
  Elkind}]{sko-yu-fal-elk:j:sc-cc}
Skowron, P.; Yu, L.; Faliszewski, P.; and Elkind, E. 2015.
\newblock The complexity of fully proportional representation for
  single-crossing electorates.
\newblock \emph{Theoretical Computer Science}, 569: 43--57.

\end{thebibliography}

\appendix
\section{Omitted Examples for Section~\ref{sec:price}}\label{app:price}

\begin{example}[\citet{LS20}]\label{ex:av}
  For each $k>1$ such that $k=s^2$ for some $s\in\mathbb{N}$, consider
  an instance $I=(C, \calA, k)$ where $\calA=(A_1, \dots, A_k)$,
  $C=\{a_1, \dots, a_k, b_1, \dots, b_{k-s}\}$,
  $A_1=\dots=A_s=\{a_1, \dots, a_k\}$ and $A_{s+i}=\{b_i\}$ for
  $i=1, \dots,k-s$.  Clearly, the committee $W = \{a_1, \dots, a_k\}$
  maximizes the utilitarian social welfare. However, the justified
  representation axiom requires that each of the voters
  $s+1, \dots, k$ has one of her approved candidates in the committee,
  so a committee that provides JR has to contain candidates
  $b_1, \dots, b_{k-s}$, and hence at most $s$ candidates from $W$.
  The utilitarian social welfare of every such committee is at most
  $k+s(s-1)=2k-\sqrt{k}$, whereas the utilitarian social welfare of
  $W$ is $sk=k\sqrt{k}$.
\end{example}

\begin{example}[\citet{LS20}]
\label{ex:coverage-EJR}
  For each even $k$, we construct an election where every EJR committee achieves at most $\nicefrac{3}{4}+O(\nicefrac{1}{k})$ 
  fraction of the optimal coverage.
  Let $k = 2t$ be the committee size and consider an election with
  candidate set $\{a_1, \ldots, a_{\nicefrac{k}{2}}\} \cup \{b_1, \ldots b_k\}$ and a collection of $n = 2k$ voters such that:
  \begin{enumerate}
      \item each of the first $k$ voters approves all of the candidates $\{a_1, \ldots, a_{\nicefrac{k}{2}}\}$, and
      \item each of the following $k$ voters approves a distinct candidate among $\{b_1, \ldots, b_k\}$.
  \end{enumerate}
  Note that $\nicefrac{n}{k} = 2$. Further,
  the first $k$ voters form a $\nicefrac{k}{2}$-cohesive group of
  size $\frac{k}{2} \cdot \frac{n}{k} = k$. As a consequence, every
  EJR committee must include all the candidates from the set
  $\{a_1, \ldots, a_{\nicefrac{k}{2}}\}$. As each of the remaining candidates covers a distinct voter, altogether every EJR committee covers exactly $k+\nicefrac{k}{2} = 1.5k$ voters. 
  On the other hand, by choosing one of the candidates from $\{a_1, \ldots, a_{\nicefrac{k}{2}}\}$ and all but one candidate from  $\{b_1,\ldots, b_k\}$, we cover $2k-1$ voters. Thus the coverage price of EJR in this election is $\nicefrac{1.5k}{2k-1} \approx \nicefrac{3}{4}$.  
  
\end{example}

\section{Addendum for Theorem~\ref{thm:overall-price-sw}}
\label{app:explicit-bound}

Theorem~\ref{thm:overall-price-sw} shows the existence of a committee size $k'$ such that  
  $\calP_\sw(k)\le \frac{1}{\beta}(1+\sqrt{k})$ for all $k\ge k'$, but does not specify an explicit value of $k'$.
Here, we show that the theorem holds whenever
$k \ge \frac{25}{\left(1 - \frac{\beta^2}{4}\right)^2}$; note that the latter quantity is at least $25$.
To see this, observe that we need two conditions for $k$ in the proof of the theorem.
The first condition, at the end of the second paragraph, is that $k \ge 2x\sqrt{k}+1$.
This holds because when $k \ge 25$, we have
\[
k \ge 4\sqrt{k} + 1 \ge 2x\sqrt{k} + 1.
\]
The second condition, denoted by inequality~\eqref{eq:sw:jr} in the proof, is that 
\[
-1 - \alpha x^2 + x +\beta + \frac{x}{\sqrt{k}}   \leq \sqrt{k}\big(\alpha x^2 - \alpha \beta x  + 1 \big).
\]
This is true because
\begin{align*}
\sqrt{k}\big(\alpha x^2 - \alpha \beta x  + 1 \big)
&\ge \sqrt{k}\left(1 - \frac{\beta^2}{4}\right) \\
&\ge 5 \\
&\ge - 1 - \alpha x^2 + 6 \\
&\ge - 1 - \alpha x^2 + x + \beta + \frac{x}{\sqrt{k}}.
\end{align*}
Here, the first inequality follows from our derivation near the end of the proof, the second inequality from $k \ge \frac{25}{\left(1 - \frac{\beta^2}{4}\right)^2}$, and the last inequality from the fact that $x$, $\beta$, and $\frac{x}{\sqrt{k}}$ are all at most $2$.

\section{Addendum for \Cref{sec:experiment}} \label{app:experiments}

\paragraph{Comparison to other experiments in the literature}
\citet{LS20} study the \emph{utilitarian guarantee} for the Chamberlin--Courant (CC) voting rule, i.e., the rule that aims to maximize coverage, as well as the \emph{representation guarantee} of Approval Voting (AV), i.e., the rule that aims to maximize social welfare, averaged over randomly generated elections. 
Though these experiments are similar to ours, we point out three main differences: 
\begin{enumerate}[(i)]
\item Both of their guarantees are worst-case bounds. Specifically, for a given instance, the utilitarian guarantee of CC is the ratio of the smallest social welfare of a CC committee and the maximum possible social welfare, and similarly for the representation guarantee of AV. This is in contrast to our experiments, where, for a fixed fraction of optimal coverage, we compute the highest social welfare that can be achieved.
\item In our experiments, we focus on committees providing JR. This is in contrast to the representation guarantee for AV, as AV does not provide JR. 
\item As \citet{LS20} focus on evaluating the utilitarian and representation guarantees of several voting rules (also beyond CC and AV), they only consider the extreme scenarios, i.e., how much social welfare (coverage) one can guarantee if one requires maximum coverage (social welfare). In contrast, we are interested in the entire Pareto curve.
\end{enumerate}

One of the experiments carried out by \citet{BFKN19} is closely related to ours. For a fixed election, they compute all combinations of social welfare and coverage achievable by committees in this election, and also indicate which of these combinations 
are achievable by a committee providing JR. In contrast, we aggregate our results over multiple elections, thus making them more robust. 

\citet{KKE+19a} carry out experiments which are conceptually similar to ours though the setting differs. More precisely, they consider committee selection under ordinal preferences, and are interested in the trade-off between the objective functions of several committee scoring rules. To this end, their paper illustrates Pareto curves for any combination of two objective functions, resulting in graphs that are similar to the ones in Figure~\ref{fig:pareto_curves1}.

\paragraph{Explanation of grid-like pattern in \Cref{fig:topIC}}

For the $\sw$-axis, we explain the effect as follows (a similar explanation works for the $\cov$-axis). Recall that the social welfare coordinate of each point is calculated as $\sw(W)/\sw(I)$, where $W$ is the output of the greedy heuristic. For each instance $I$, $\sw(W)$ can be written as $\sw(I) - \Delta_W$, where $\Delta_W \in \mathbb{N}$ is the loss of the greedy heuristic. Our first observation is that for the $0.1$-IC model, the value $\sw(I)$ does not change much across different instances. As a result $(\sw(I) - \Delta_W)/\sw(I)$ is very close to $(\sw(I') - \Delta_{W'})/(\sw(I'))$ for another instance $I'$ whenever $\Delta_W = \Delta_{W'}$. Otherwise, the two values are separated by a gap. Hence, each of the layers depicts instances with equal $\Delta_W$. This effect is strongest when $\Delta_W = 0$ and diminishes as $\Delta_W$ becomes larger.

\paragraph{Computational Infrastructure} The experiments were performed on a system with 3.4 GHz Quad-Core Intel Core i5 CPU, 8GB RAM,
and macOS 11.4 operating system. 
Python 3.8.8 was used for implementation and the libraries matplotlib 3.3.4, numpy 1.20.1, and pandas 1.2.4 were used. For solving integer programs we used gurobi 9.1.2.

\end{document}